\documentclass[12pt,onecolumn,draftclsnofoot,journal]{IEEEtran}
\usepackage{graphicx}
\usepackage{float}
\usepackage{epstopdf}
\usepackage[cmex10]{amsmath}
\usepackage{array}
\usepackage{cite}
\usepackage{amssymb}
\usepackage{amsfonts}
\usepackage{amsmath}
\usepackage{stackrel}
\usepackage{booktabs,multirow}
\usepackage{arydshln}
\usepackage{slashbox}
\usepackage{amsthm}
\usepackage{listings}
\usepackage{algorithm}
\usepackage{algorithmicx}
\usepackage{algpseudocode}
\usepackage{threeparttable}
\newcommand\numberthis{\addtocounter{equation}{1}\tag{\theequation}}

\newtheorem{lem}{Lemma}
\newtheorem{rem}{Remark}
\newtheorem{theo}{Theorem}

\newtheorem{cor}{Corollary}
\newtheorem{pro}{Proposition}

\newtheorem{ex}{Example}
\newfloat{routine}{htbp}{loa}
\floatname{routine}{Routine}

\makeatletter
\newcommand{\algmargin}{\the\ALG@thistlm}
\makeatother
\newlength{\forwidth}
\algdef{SE}[parFOR]{parFor}{EndparFor}[1]
  {\parbox[t]{\dimexpr\linewidth-\algmargin}{
     \hangindent\forwidth\strut\algorithmicfor\ #1\ \algorithmicdo\strut}}{\algorithmicend\ \algorithmicfor}
\algnewcommand{\parState}[1]{\State
  \parbox[t]{\dimexpr\linewidth-\algmargin}{\strut #1\strut}}

\newlength{\ifwidth}
\algdef{SE}[parIF]{parIf}{EndparIf}[1]
  {\parbox[t]{\dimexpr\linewidth-\algmargin}{
     \hangindent\ifwidth\strut\algorithmicif\ #1\ \algorithmicdo\strut}}{\algorithmicend\ \algorithmicif}
\begin{document}

\title{Asymptotic Average Multiplicity of Structures within Different Categories of Trapping Sets, Absorbing Sets and Stopping Sets in Random Regular and Irregular LDPC Code Ensembles}
\author{Ali Dehghan, and Amir H. Banihashemi,\IEEEmembership{ Senior Member, IEEE}
\thanks{Part of this work is accepted for presentation at {\em ISIT 2018}, Vail, Colorado, USA, June 2018.}}
\maketitle

%\IEEEpeerreviewmaketitle

\begin{abstract}
The performance of low-density parity-check (LDPC) codes in the error floor region
is closely related to some substructures of the code's Tanner graph, collectively referred to as
{\it trapping sets (TSs)}. In this paper, we study the asymptotic average number of different types of trapping sets such as {\em elementary TSs (ETS)}, {\em leafless ETSs (LETS)},
{\em absorbing sets (ABS)}, {\em elementary ABSs (EABS)}, and {\em stopping sets (SS)},
in random variable-regular and irregular LDPC code ensembles. We demonstrate that, regardless of the type of the TS, as the code's length tends to infinity, the average number of a given structure tends to infinity,
to a positive constant, or to zero, if the structure contains no cycle, only one cycle, or more than one cycle, respectively. 
For the case where the structure contains a single cycle, we derive the asymptotic expected multiplicity of the structure by counting the average number of its 
constituent cycles and all the possible ways that the structure can be constructed from the cycle. This, in general, involves computing the expected number of cycles of a certain length with a certain given combination of node degrees,
or computing the expected number of cycles of a certain length expanded to the desired structure by the connection of trees to its nodes. 
%of different length consisting of variable nodes of different degrees through the available approximations for the average number of its constituent cycle. 
The asymptotic results obtained in this work, which are independent of the block length and only depend on the code's degree distributions, are shown to be accurate even for finite-length codes.

\begin{flushleft}

\end{flushleft}

\end{abstract}

\section{introduction}

The performance of low-density parity-check (LDPC) codes in the error floor region is closely related to some substructures of the code's Tanner graph, here referred to collectively as {\em trapping sets (TS)}.
A trapping set ${\cal S}$ is often identified by its number of variable nodes $a$, and the number of unsatisfied check nodes $b$ in its subgraph.\footnote{The term ``unsatisfied'' is used for check nodes that are connected to an odd number of variable nodes in the trapping set. This term (in the case of binary LDPC codes) implies that if the variable nodes in the trapping set ${\cal S}$ are set to one, and all the other variable nodes in the Tanner graph are set to zero, then 
these check nodes will be unsatisfied, i.e., the modulo-2 addition of the variables connected to the check node is non-zero.} In this case, the set ${\cal S}$ is said to belong to the $(a,b)$ class.
For a given LDPC code, the harmful trapping sets would depend on the degree distributions of the code, the channel model, the decoding algorithm, the quantization scheme, and the structure of the code's Tanner graph. Some
categorizations of TSs include {\em stopping sets (SS)}, {\em elementary TSs (ETS)}, {\em leafless ETSs (LETS)}, {\em absorbing sets (ABS)}, and {\em elementary ABSs (EABS)}.
Stopping sets are known to be the problematic structures of the belief propagation algorithm over the binary erasure channel (BEC)~\cite{Di}. Leafless ETSs are relevant for
soft-decision iterative decoding of LDPC codes over the additive white Gaussian noise (AWGN) channel and have been widely studied in the literature, 
see, e.g.,~\cite{MR2991821},~\cite{KB},~\cite{hashemi2015new}, and the references therein. In fact, LETS structures appear to form an overwhelming majority of dominant trapping sets of 
variable-regular LDPC codes over the AWGN channel~\cite{hashemi2015new}. 
The broader category of ETSs are harmful sets for irregular codes over the AWGN channel~\cite{MR2991821},~\cite{Y-Arxiv}. Absorbing sets are the fixed points of bit-flipping decoding algorithms~\cite{xiao2007estimation},~\cite{dolecek2010analysis},
and are also shown to be relevant in the context of quantized decoders over the AWGN channel~\cite{dolecek2010analysis}.

Regardless of the type of the TS, the problem of counting and/or enumerating trapping sets in a given code is a hard problem.
%(i.e. there is no polynomial time algorithm for it).
It was shown in \cite{mcgregor2010hardness} that for a given $a$, finding an $(a,b)$ TS with the smallest $b$ is an NP-hard problem. Also, finding an $(a,b)$ ETS with the smallest $a$ for a given $b$ is a hard problem~\cite{mcgregor2010hardness}.
Very recently, similar hardness results were proved for LETSs and EABSs~\cite{dehghan2018hardness}.
Furthermore, it was proved in \cite{krishnan2007computing} that for a given Tanner graph $G$ and a positive
integer $t$, determining whether $G$ has a stopping set of size $t$ is NP-complete.
Despite the intrinsic difficulty of such combinatorial problems, many advances are made on characterization of trapping sets and the development of efficient search algorithms to find them~\cite{MR2582870},~\cite{rosnes2009efficient},~\cite{zhang2011efficient},  \cite{rosnes2012addendum},~\cite{kyung2012finding},~\cite{MR2991821}, \cite{rosnes2014minimum}, \cite{KB}, \cite{hashemi2015characterization}, \cite{hashemi2015new},~\cite{Y-Arxiv}.

More related to the work presented in this paper, asymptotic results on the distribution of trapping sets in different ensembles of LDPC codes were established in~\cite{BM},~\cite{orlitsky2005stopping},~\cite{MR2292873},~\cite{MR2810289},~
\cite{amiri2013asymptotic}.
The distribution of stopping sets in Tanner graph ensembles was first studied in \cite{BM}. In~\cite{orlitsky2005stopping}, using tedious combinatorial methods, Orlitsky {\it et al.} studied the asymptotic average distribution of stopping sets and the stopping number (the size of the smallest nonempty stopping set in a code) in regular and irregular LDPC code ensembles, as the code's block length $n$ tends to infinity.
In particular, they proved that for almost all codes with smallest variable degree greater than  two, stopping number increases linearly with $n$. Milenkovic {\it et al.}~\cite{MR2292873}
used random matrix enumeration techniques and large deviations theory to study the
asymptotic distribution of SSs and TSs of size constant and linear in $n$, in random
regular and irregular LDPC code ensembles. Specifically, for regular ensembles with variable degree $d_v$, they demonstrated that as $n$ tends to infinity, for constant $a$ and $b$ values,
the average number of ETSs in the $(a,b)$ class tends to
$c n^{a+\frac{b-ad_v}{2}}$, where $c$  is a constant with respect to $n$, but depends on $a$, $b$, and $d_v$.
Abu-Surra {\em et al.}~\cite{MR2810289} 
studied the ensemble stopping set and trapping set enumerators for protograph-based LDPC codes both in finite-length and asymptotic regimes. The focus in \cite{MR2810289}
was to determine whether or not the typical relative smallest size of trapping or stopping sets grows linearly with the block length.  
Using a technique similar to that of~\cite{MR2292873}, Amiri {\it et al.}~\cite{amiri2013asymptotic} studied the asymptotic distribution of ABSs and fully ABSs (FABSs) of size linear in $n$, in random
regular LDPC code ensembles. They also derived simplified formulas for enumerating the asymptotic average number of EABSs, and FEABSs in regular LDPC code ensembles with variable degrees $3$ and $4$.

In this paper, we focus on the case where $a$ and $b$ are constant values (with respect to $n$). This case is of particular practical interest, as in practice, the size of the most harmful trapping sets may not necessarily
increase with the code's block length.
For the scenario of constant $a$ and $b$ values, the limited existing results on the asymptotic average number of trapping sets, as discussed in the previous paragraph, are consistent with the results of our analysis.
In addition, for this scenario, our analysis both extends and deepens the results of the literature. For example, while the results in~\cite{MR2292873} are limited to TSs and ETSs, in this work, we also cover LETSs.
As an example of how the results of this work are more accurate and provide a deeper understanding about the TS structures, we consider the only result of~\cite{MR2292873} on $(a,b)$ ETSs with constant $a$ and $b$ values. This result, which is limited to biregular ensembles, indicate that the asymptotic expected number of $(a,b)$ ETSs is equal to $c n^{a+\frac{b-ad_v}{2}}$, for some constant $c$. For the interesting case where $b=a(d_v-2)$, this result implies that, on average, there are a constant (non-zero) number of ETSs within the $(a,b)$ class under consideration. In this work, we determine the exact value of this constant as a function of the graph's degree distributions and the values of $a$ and $b$. In addition, we obtain similar results for variable-regular and irregular ensembles.  
 
Another major difference in comparison with the existing literature~\cite{orlitsky2005stopping}, \cite{MR2292873}, \cite{amiri2013asymptotic}, is that rather than focusing on a class of trapping sets, 
we focus on individual non-isomorphic structures within a class. We show that
such structures, although in the same class, may demonstrate completely different asymptotic behavior. For example, the average number of one structure may tend to a constant as $n$ tends to infinity, while that of another structure, within the same class, may tend to infinity. Or the asymptotic average number of one structure may be zero, while that of another structure may be a non-zero constant. If one only considers the behavior of the whole class, in the former case, the average number tends to infinity and in the latter, to a constant. The granularity of our analysis is important since it is known that, in general, different non-isomorphic trapping sets in the same class can differ in their harmfulness~\cite{BS}. It is thus important
to know the average number of individual non-isomorphic TS structures within a code ensemble.  Finally, although the analysis in~\cite{orlitsky2005stopping},~\cite{MR2292873},~\cite{amiri2013asymptotic}, may be conceptually easy to understand, the derivations are tedious and the final formulas are often complicated. In contrast, at the core of our analysis is a simple, yet general, asymptotic result that the expected number of an structure ${\cal S}$ is equal to
$\Theta(n^{|V(\mathcal{S})|-|E(\mathcal{S})|})$,\footnote{We use the notation $f(x)=\Theta(g(x))$, if for sufficiently large values of $x$, we have $a \times g(x) \leq f(x) \leq b \times g(x)$, for some positive $a$ and $b$ values.}
where $|V(\mathcal{S})|$ and $|E(\mathcal{S})|$ are the number of vertices (nodes) and the number of edges of ${\cal S}$, respectively. This result is then easily translated to an asymptotic result on the average multiplicity of the structure
based on the number of cycles in the structure. For the number of cycles equal to zero, one, or more than one, the average multiplicity tends to infinity, a non-zero constant, or zero, respectively.

In an ensemble ${\cal E}$ of LDPC codes, for given $a$ and $b$ values, we say that the $(a,b)$ class has a {\it consistent behavior} in ${\cal E}$, if all the structures within the class demonstrate the same asymptotic behavior in ${\cal E}$, i.e., the average number of every non-isomorphic structure in the class tends to the same value (zero, infinity or a non-zero constant) in ${\cal E}$, as $n$ tends to infinity.
Otherwise, we say that the class has an {\it inconsistent behavior} in ${\cal E}$.
As an example of our results, we prove that in random variable-regular LDPC codes, every class of ETSs and every class of LETSs has a consistent behavior. In
irregular LDPC codes, however, classes of ETSs and LETSs have inconsistent behavior, in general.
%Another consequence of our analysis is to provide approximations for the average number of TS structures that contain only one cycle. By definition, such a cycle must be chordless (simple). We use asymptotic results on the multiplicity of cycles in random regular and irregular Tanner graphs~\cite{dehghan2016new} to approximate the average number of simple cycles and thus the average number of corresponding structures in the ensemble. We demonstrate that these approximations are rather accurate even for individual codes selected randomly from finite-length ensembles.

In Table \ref{my-label}, we have summarized our results on the asymptotic average number of $(a,b)$ TSs in variable-regular ensembles of LDPC codes
with variable degree $d_v$, for different types of trapping sets, and for different values of $a$, $b$ and $d_v$. As we will discuss later, the asymptotic expected multiplicity of trapping sets depends on the value of $b/a$ in relation to $d_v-2$.
This can be seen in Table~\ref{my-label} through the distinction of different columns. 
In Table \ref{my-label}, the notation $N_{2a}$ denotes the asymptotic average number of cycles of length $2a$ in the Tanner graphs of the codes. The entry ``--'' in the table is used to
indicate that it is impossible to have trapping sets of a particular type in the corresponding classes. In Table \ref{my-label}, the value $N_{2a}$ is used as an asymptotic approximation for the average number of 
TS structures that are isomorphic to a simple cycle of length $2a$. The value $N_{2a}$ was computed (approximated) recently in~\cite{dehghan2016new} for random regular (irregular) Tanner graphs.

As another contribution of this work, as related to variable-regular graphs (Table~\ref{my-label}), we focus on structures which are not simple cycles but contain only a single (simple) cycle. These are the remaining structures (in addition to simple cycles) whose asymptotic expected multiplicity, we have proved to be a non-zero constant. For such structures, we first characterize them as a simple cycle appended by a number of trees of different sizes rooted at the variable nodes of the cycle. We then use this characterization to count the number of such structures on average in the asymptotic regime. The counting problem in this case is formulated recursively, and has a solution which is a generalization of Catalan numbers \cite{MR1098222}.
These results correspond to the entry $\Theta(1)$ in Table~ \ref{my-label}. 
%compute the average multiplicity of ETSs in variable-regular LDPC codes in classes with $b/a = d_v -2$. This corresponds to the entry $\Theta(1)$ in Table~ \ref{my-label}. 

We also investigate the asymptotic expected multiplicity of different types of TSs in irregular graphs, in cases where such expected values are non-zero and finite. In particular, to compute the asymptotic average multiplicity of LETS, ABS and SS structures 
in irregular graphs, we generalize the results of \cite{dehghan2016new} to count the number of cycles of different lengths with different combinations of node degrees within irregular ensembles of bipartite graphs. Moreover, to compute the asymptotic average multiplicity of ETS structures, we generalize the results derived for regular graphs (recursively counting the single cycles appended by trees) to irregular graphs.

%Moreover, we generalize this result to compute the expected number of ETSs of irregular LDPC codes in cases where such an  expected number is constant. Similar results are derived for computing the average multiplicity of LETS, ABS and SS structures in irregular codes in cases where the multiplicity is constant (non-zero and finite). For such 

\begin{table}[]
\centering
\caption{The asymptotic average number of different trapping set structures within $(a,b)$ classes of variable-regular LDPC code ensembles with variable degree $d_v$.}
\label{my-label}
\begin{tabular}{|l||c|c|c|c|c|}
\hline
 TS Type       & $b/a< d_v -2 $ & \multicolumn{3}{c|}{$b/a= d_v -2 $}       & $b/a> d_v -2 $ \\ \hline
       &    $d_v\geq 2$ &     $d_v=2$   & $d_v=3$   & $d_v\geq 4$   &  $d_v\geq 2$   \\ \hline
  ETSs     &       0     & \multicolumn{3}{c|}{$\Theta(1)$}        &  $\infty$  \\ \hline
  LETSs    &       0     & \multicolumn{3}{c|}{$\sim N_{2a}$}                &  -- \\ \hline
  EASs     &     0      & \multicolumn{2}{c|}{$\sim N_{2a}$} &    --         &   -- \\ \hline
  SSs    &      0      & $\sim N_{2a}$ & \multicolumn{2}{c|}{--}    &   -- \\ \hline
% SSs of size $a$&  \multicolumn{5}{c|}{$\leq N_{2a}$}       \\ \hline
\end{tabular}
%  \begin{tablenotes}
%    \item[1] $+ $ If there are $\Theta(n)$ check nodes such that the degree of each of them is greater than   or equal to the maximum check degree of any member of $(a,b)$ ETSs.
%  \end{tablenotes}
\end{table}

The organization of the rest of the paper is as follows: In Section~\ref{S2}, we present some definitions and notations. This is followed in Section~\ref{S3} by our main result
on the asymptotic average number of an arbitrary local structure in the Tanner graph of random LDPC codes. In this section, we also apply
the general result to different types of trapping sets including ETSs, LETSs, ABSs, EABSs and SSs,
for both variable-regular and irregular LDPC code ensembles. In Section~\ref{LETS-IRREG}, we compute the expected number of 
SS, ABS and LETS structures whose asymptotic average multiplicity is constant in irregular LDPC codes. Next, In Section \ref{sec-ets}, we compute the expected number of 
ETS structures whose asymptotic average multiplicity is constant in biregular, variable-regular and irregular ensembles.
Section~\ref{S5} is devoted to numerical results. The paper is concluded with some remarks in Section~\ref{S7}.

\section{Definitions and notations}
\label{S2}

For a graph $G$, we denote the node set and the edge set of $G$ by $V(G)$ and $E(G)$,
respectively.
%To avoid ambiguity between the edge set and the expected value, throughout the work we denote the expected value by $ \mathbf{E}(.)$.
In this work, we consider graphs with no loop or parallel edges, where a loop is defined as an edge that connects a node to itself.
%The number of nodes of $G$, $|V(G)|$, is called the {\it{order}} of $G$.
%Similarly, we denote the number of edges by $|E(G)|$.
For $v \in V(G)$ and $S\subseteq V(G)$, notations $N(v)$ and $N(S)$ are used to denote
the neighbor set of $v$, and the set of nodes of $G$ which has a neighbor in $S$, respectively.
A {\it path} of length $c$ in $G$ is a sequence of distinct nodes $v_1, v_2, \ldots , v_{c+1}$ in $V(G)$, such that $\{v_i, v_{i+1}\} \in E(G)$, for $1 \leq i \leq c$.
A {\it cycle} of length $c$ is a sequence of distinct  nodes
$v_1, v_2, \ldots , v_{c}$ in $V(G)$ such that $v_1, v_2, \ldots , v_{c}$ form a path of length $c-1$, and
$\{v_c, v_{1}\} \in E(G)$. We may refer to a path or a cycle by the set of their nodes, or by the set of their edges, or by both.
A {\em chordless} or {\em simple} cycle in a graph  is a cycle such that no two nodes of the cycle are connected by an edge that does not itself belong to the cycle.

A graph $G=(V,E)$ is called {\it bipartite}, if the node set $V(G)$ can be
partitioned into two disjoint subsets $U$ and $W$ (i.e., $V(G) = U \cup W \text{ and } U \cap W =\emptyset $), such that every edge in $E$ connects a node
from $U$ to a node from $W$. The {\it Tanner graph} of a low-density parity-check (LDPC) code is a bipartite graph, in
which $U$ and $W$ are referred to as {\it variable nodes} and {\it check nodes}, respectively. Parameters $n$ and $n'$ in this case are used to denote $|U|$ and $|W|$, respectively. Parameter $n$ is the  length of the code,
and the code rate $R$ satisfies $R \geq 1- (n'/n)$.

The number of edges incident to a node $v$ is called the {\em degree} of $v$, and is denoted by $d_v$ or $d(v)$.
Also, we denote the {\it{maximum degree}} and the {\it{minimum degree}} of $G$ by $\Delta(G)$ and $\delta(G)$, respectively.
A bipartite graph $G = (U\cup W,E)$ is called {\it biregular}, if all the nodes on the same side of the given bipartition have the same degree,
i.e., if all the nodes in $U$ have the same degree $d_v$ and all the nodes in $W$ have the same degree $d_c$.
It is clear that, for a biregular graph, $|U|d_v=|W|d_c=|E(G)|$. A bipartite graph that is not biregular is called {\it irregular}.
A Tanner graph $G = (U\cup W,E)$ is called variable-regular with variable degree $d_v$,
if for each variable node $u_i\in U$, $d_{u_i} = d_v$. Also, a $(d_v, d_c)$-regular Tanner graph is a variable-regular graph with variable degree $d_v$, in which
for each check node $w_i \in W$, $d_{w_i} = d_c$. The variable and check node degree distributions of irregular LDPC codes are represented by polynomials $\lambda(x) = \sum_{i=d_{v_{\min}}}^{d_{v_{\max}}} \lambda_i x^{i-1}$ and
$\rho(x) = \sum_{i=d_{c_{\min}}}^{d_{c_{\max}}} \rho_i x^{i-1}$, where $\lambda_i$ and $\rho_i$ denote the fraction of the edges in the Tanner graph connected to degree-$i$ variable and check nodes, respectively, and
$d_{v_{\min}}$, $d_{v_{\max}}$, $d_{c_{\min}}$ and $d_{c_{\max}}$ are the minimum and maximum variable and check node degrees in the Tanner graph, respectively.

In a Tanner graph $G = (U\cup W,E)$, for a subset $S$ of $U$, the induced subgraph of $S$ in $G$, denoted by $G(S)$, is the graph with the set of nodes
$S\cup N(S)$, and the set of edges $\{\{u_i,w_j\}: \{u_i,w_j\} \in E(G) , u_i\in S , w_j \in N(S)\}$.
The set of check nodes with odd and even degrees in $G(S)$ are denoted by $N_o(S)$ and $N_e(S)$, respectively.
Also, the terms {\em unsatisfied} check nodes and {\em satisfied} check nodes are used to refer to the check nodes in $N_o(S)$ and $N_e(S)$, respectively.
Throughout this paper, the size of an induced subgraph $G(S)$ is defined to be the number of its variable nodes (i.e., $|S|$).

For a given Tanner graph $G$, a set $S \subset U$, is said to be an {\it $(a,b)$ trapping set (TS)} if $|S| = a$ and
$|N_o(S)| = b$. Alternatively, set $S$ is said to belong to the {\em class} of $(a,b)$ TSs.
A {\em stopping set (SS)} $S$ is a TS for which the degree of each check node in $G(S)$ is at least two.
An {\it elementary trapping set
(ETS)} is a TS for which all the check nodes in $G(S)$ have degree one or two.
A {\it leafless ETS (LETS)} $S$ is an ETS for which each variable node in $S$ is connected to at least two satisfied check nodes in $G(S)$.
An {\em absorbing set (ABS)} $S$ is a TS for which all the variable nodes in $S$
are connected to more nodes in $N_e(S)$ than in $N_o(S)$.
Also, an {\it elementary absorbing set (EABS)} $S$ is an ABS for which all the check nodes in $G(S)$ have degree one or two.
A {\it fully absorbing set  (FABS)} is an ABS with the extra constraint that each variable node in $U \setminus S$ is connected to strictly fewer nodes in  $N_o(S)$ than
in $W \setminus N_o(S)$. Also, an {\it elementary fully absorbing set  (EFABS)} $S$ is an FABS such that each check node  in $G(S)$ has degree one or two.

In the asymptotic analysis presented in this work, for a given ensemble of LDPC codes (Tanner graphs), identified by certain degree distributions and the block length $n$, we consider the case where $n$ tends to infinity.
In such an asymptotic scenario, we say that a structure ${\cal S}$ is {\it local} in a Tanner graph $G$, if the definition of ${\cal S}$ depends on a constant number (with respect to $n$) of nodes in $G$.
In particular, as we are interested in subgraphs ${\cal S} = G(S)$ induced in $G$ by a set of variable nodes $S$, if such subgraphs are local, we refer to them as being {\em locally induced} subgraphs.
For example, if one considers a constant positive integer $a$, and Tanner graphs whose maximum variable degree is a constant in $n$, then $(a,b)$ ETSs, $(a,b)$ LETSs, $(a,b)$ ABSs, $(a,b)$ EABSs, and $(a,b)$ SSs are all local structures. (Note that such sets can only exist for $b \leq a \times d_{v_{\max}}$, and thus $b$ is also a constant in $n$.)
On the other hand, the definition of FABSs depends on all the variable nodes in the Tanner graph, and thus FABSs are not local structures.
In this work, our focus is on $(a,b)$ TSs with $a$ and $b$ being constants in $n$. We also consider
Tanner graphs whose maximum variable and check node degrees are constant. Thus, in this work, all ETSs, LETSs, ABSs, EABSs and SSs are local structures. We note that
the locality constraint is one of the main assumptions required for the proof of our results including Theorem~\ref{T1} in the next section.

We conclude this section with the definition of $k$-permutations that will be used in the proof of our main result. The {\it $k$-permutations} of $n$ objects are the different ordered arrangements of $k$-element subsets of the $n$ objects. The number of such permutations is equal to $P(n,k)=n(n-1)\cdots(n-k+1)$.

\section{Asymptotic Average Number of an Arbitrary Local Structure in LDPC Code (Tanner Graph) Ensembles}
\label{S3}

\subsection{Main Result}

In this subsection, we find a lower and an upper bound on the average number of a locally induced subgraph of a random Tanner graph with a given
degree distribution, in the asymptotic regime where the size of the graph tends to infinity.

\begin{theo}\label{T1}
Let $\Delta$ be a fixed natural number, such that $\Delta \geq d_1 \geq d_2  \geq \ldots \geq d_n$, and $\Delta \geq d_1' \geq d_2'  \geq \ldots \geq d_{n'}'$, and that $\sum_{i=1}^n d_i =\sum_{i=1}^{n'} d_i'=\eta$.
Consider the probability space $\mathcal{G}$ of all Tanner graphs with node set $(U,W)$, where $U=\{u_1, u_2, \ldots , u_n\}, W=\{w_1,w_2, \ldots, w_{n'}\}$, and $d(u_i)=d_i$,
$d(w_i)=d'_i$. Suppose that the degree sets $\{d_i\}$ and $\{d'_i\}$ are selected according to the distributions $\lambda(x)$ and $\rho(x)$, respectively, and that
the graphs in $\mathcal{G}$ are selected uniformly at random.
For $G \in \mathcal{G}$, denote by $X_{\mathcal{S}}(G)$ the number of copies of a subgraph $\mathcal{S}$, induced by a constant number of variable nodes, in $G$.
Then, in the asymptotic regime of $n \rightarrow \infty$,
we have
\begin{equation}\label{E0}
\mathbf{E}(X_{\mathcal{S}}) = \Theta(n^{|V(\mathcal{S})|-|E(\mathcal{S})|})\:,
\end{equation}
where $\mathbf{E}(X_{\mathcal{S}})$ is the expected value of  $X_{\mathcal{S}}(G)$, and $|V(\mathcal{S})|$ and $|E(\mathcal{S})|$ are the number of nodes and edges of $\mathcal{S}$, respectively.
\end{theo}

\begin{proof}{
First, we introduce the method that we use to construct a random Tanner graph. The method is the same as the one used in \cite{dehghan2016new}.
Suppose that the variable nodes are labeled as $\{u_1, u_2, \ldots , u_n\}$, and check nodes as
$\{w_1,w_2, \ldots, w_{n'}\}$. For each node $z$ with degree $d(z)$, we assign a bin that contains $d(z)$ cells, and consider a
random perfect matching (bijection) to pair the cells on one side of the graph to the cells on the other side. The bipartite graphs constructed by connecting any two paired cells with an edge are called  {\em configurations}.
%In fact we consider a bijection between two sets of cells.
There are $\eta!$ configurations, where $\eta$ is the number of edges in the graph. In the rest of the proof, we assume that configurations are selected uniformly at random.
From each configuration,  we construct a Tanner graph such that if there is an edge between two cells, then we place an edge between the corresponding nodes (bins) in the Tanner graph.
The resulted Tanner graphs can thus be considered as images of the configurations that are obtained from the random perfect matchings.
Following the notation in \cite{dehghan2016new}, we denote the ensemble of bipartite graphs so constructed by $\mathcal{G}^*$, and note that the graphs in $\mathcal{G}^*$ can have parallel edges.
We also note that a uniform distribution over the configurations induces a non-uniform distribution over $\mathcal{G}^*$. The next step in the construction of $\mathcal{G}$ is to remove all the graphs with 
parallel edges from $\mathcal{G}^*$. It is now easy to see that a uniform distribution over the configurations, with the condition that images with parallel edges are rejected, induces a uniform distribution over  $\mathcal{G}$.
The reason is that corresponding to each graph in $\mathcal{G}$, we have the same number $(d_1!)(d_2!)\ldots (d_n!)(d_1'!)( d_2'!)\ldots( d_{n'}'!)$ of 
configurations.

It was proved in \cite{dehghan2016new} that rather than working in the probability space $\mathcal{G}$, one can work in $\mathcal{G}^*$, and that if a property holds true asymptotically almost surely (a.a.s.) for $\mathcal{G}^*$,
it also holds true a.a.s. for $\mathcal{G}$. (See the proof of Theorem~1 in \cite{dehghan2016new}.) We adopt the same approach here, and in addition, instead of directly working in $\mathcal{G}^*$, which has a non-uniform distribution, we
perform the calculations in the space of configurations (with uniform distribution). 

%There are $N_0(\eta)=\eta!$ configurations, where $\eta$ is the number of edges in the graph. Also, note that  given a set of $\ell$ fixed edges, there are $N_{\ell}(\eta)=(\eta-\ell)!$ configurations containing those edges. Since each bipartite graph (which  has no parallel edges) corresponds to $(d_1!)(d_2!)\ldots (d_n!)(d_1'!)( d_2'!)\ldots( d_{n'}'!) $ matchings, a bipartite graph can be chosen uniformly at random by choosing a matching uniformly at random and rejecting the result if it has parallel edges. We thus consider perfect matchings with no parallel edges in their image that are selected uniformly at random. Given a set of $\ell$ fixed edges, we use the notation $M_{\ell}(\eta)$ for the number of such configurations containing those edges. We then have: $M_{\ell}(\eta)/ M_{0}(\eta) \sim N_{\ell}(\eta)/ N(\eta)$. (See, Corollary 2.18 of \cite{MR1864966}.)

Consider a bipartite structure of interest $\mathcal{S}=(X\cup Y,E')$, such that $X=\{x_1,x_2, \ldots, x_r\} \subset U$,
$Y=\{y_1,y_2, \ldots, y_{r'}\} \subset W$, and $r+r'=|V(\mathcal{S})|$. We say there is a copy of $\mathcal{S}$ in a configuration corresponding to a Tanner graph $G=(U \cup W,E)$, if
there is a set of $|E(\mathcal{S})|$ edges in the configuration, whose image in $G$ corresponds to a subgraph which is isomorphic to  $\mathcal{S}$.
We denote the number of copies of $\mathcal{S}$ in a configuration by $C_{\mathcal{S}}$. Now considering that, given a set of  $|E(\mathcal{S})|$ edges, there are $(\eta- |E(\mathcal{S})|)!$ configurations containing those edges, and that
configurations are selected uniformly at random, we have
\begin{equation}\label{E142}
\mathbf{E}(X_{\mathcal{S}}) %= C_{\cal{S}} \dfrac{M_{|E(\mathcal{S})|}(\eta)}{ M_{0}(\eta) }\sim C_{\mathcal{S}} \dfrac{ N_{|E(\mathcal{S})|}(\eta)}{ N_0(\eta)} = 
\sim C_{\mathcal{S}}\dfrac{(\eta - |E(\mathcal{S})|)!}{\eta!} \sim \dfrac{C_{\mathcal{S}}}{\eta^{|E(\mathcal{S})|}}\:,
\end{equation}
%where $N_{\cal{S}}$ is the number of $\mathcal{S}$ structures in a Tanner graph. For the first asymptotic equality, see, Corollary 2.18 of \cite{MR1864966}.
%$(\eta - |E(\mathcal{S})|)!$ is the number of configurations containing $|E(\mathcal{S})|$ fixed edges, and
where, the last equation is valid asymptotically since $|E(\mathcal{S})|$ is a constant in $n$, and thus in the number of edges  $\eta$ of the Tanner graph. In the following,  we derive upper and lower bounds on $C_{\mathcal{S}}$, and subsequently on $\mathbf{E}(X_{\mathcal{S}})$.

\noindent
\underline{Upper bound on $C_{\mathcal{S}}$:} To form a copy of the structure $\mathcal{S}$, we choose an $ r$-permutation of the $n$ bins
from $U$ and an $ r'$-permutation of the $n'$ bins from $W$ (note that some of the permutations may result in the same copy of ${\cal S}$. By considering all possible $ r$-permutations and
$r'$-permutations, we obtain an upper bound on $C_{\mathcal{S}}$). Next, after fixing the $r$-permutation and the $ r'$-permutation, for each $i$, $1 \leq i \leq r$, for the $i^{th}$ bin in the $r$-permutation,
we select $d(x_i)$ cells in order. We note that if the number of cells in the $i^{th}$ bin is not equal to $d(x_i)$, then a copy of $\mathcal{S}$ cannot be formed (by definition, the degree
of a variable node in a subgraph induced by a set of variable nodes must be equal to the degree of that variable node in the Tanner graph). On the check side, for each $i$, $1 \leq i \leq r'$, for the $i^{th}$ bin in the $r'$-permutation,
we select $d(y_i)$ cells in order. If the number of cells in the $i^{th}$ bin is less than $d(y_i)$, then a copy of $\mathcal{S}$ cannot be formed.
(Note that by considering all possible orderings for $d(x_i)$ cells and  $d(y_i)$ cells, we find an upper bound on $ C_{\mathcal{S}}$, since some of those orderings may result in the same copy of ${\cal S}$).
The number of possible choices of the cells on the variable side is thus upper bounded by $\prod_{j=1}^{r} \prod_{i=0}^{d(x_j)-1} (d(x_j) - i)$, that can be further bounded from above by $\Delta^{|E( \mathcal{S})|}$. Similarly, the number of choices of the cells on the check side is upper bounded by $\Delta^{|E( \mathcal{S})|}$. We therefore have

\begin{align*}
C_{\mathcal{S}}
&\leq   {\Delta}^{2|E( \mathcal{S})|} \times P(n,r)\times P(n',r')\\
&\leq   {\Delta}^{2|E( \mathcal{S})|} \times n^r \times n'^{r'}\\
&\leq   {\Delta}^{2|E( \mathcal{S})|} \times n^r \times n^{r'}\\
&=      {\Delta}^{2|E( \mathcal{S})|} \times n^{|V(\mathcal{S})|}\:. \numberthis \label{AA1}
\end{align*}

\noindent
\underline{Lower bound on $C_{\mathcal{S}}$:}
If any of the variable degrees of $\mathcal{S}$ is missing in the variable node degree distribution $\lambda(x)$ of the Tanner graph $G$, then, no copy of $\mathcal{S}$ can exist in $G$. Otherwise, assume that
there are $\alpha$ different variable degrees $d_1, \ldots, d_{\alpha}$, in $\mathcal{S}$, and denote the number of variable nodes with degree $d_i$ by $r_i$. We thus have $r_1+\cdots+r_{\alpha}=r$.
Denote by $U_i$ the set of variable nodes in $G$ with degree $d_i, \:i=1,\ldots,\alpha$. By the construction of the ensemble, we have $|U_i| \geq c_i \times n$, for some constant values $c_i, \: i=1,\ldots,\alpha$.
On the check side, denote the set of all check nodes with degree at least equal to the largest check degree in $\mathcal{S}$ by $W'$. Here also, $|W'| \geq c' \times n' \geq c'' \times n$, for some constants $c'$ and $c''$.
It is then easy to establish the following lower bound on $C_{\mathcal{S}}$:

\begin{equation}
C_{\mathcal{S}} \geq {{|U_1|} \choose {r_1}} \cdots {{|U_{\alpha}|} \choose {r_{\alpha}}} {{|W'|}\choose {r'}}\:,
\label{AB0}
\end{equation}
where ${a \choose b} = a!/(b! (a-b)!)$.
Using the linear lower bounds on $|U_i|, i = 1,\ldots,\alpha$, and $|W'|$, followed by the well-known lower bound ${{a} \choose {b}} \geq (\dfrac{a}{b})^b$, we obtain
%. On the other hand,  $|U'|\geq f(n) $ and $|W'|\geq f(n) $. Thus,  by (\ref{AB0}), we have:

\begin{align*}
C_{\mathcal{S}}
&\geq \Big(\dfrac{c_1 n}{r_1}\Big)^{r_1} \cdots \Big(\dfrac{c_{\alpha} n}{r_{\alpha}}\Big)^{r_{\alpha}} \Big(\dfrac{c'' n}{r'}\Big)^{r'}\\
&\geq  \dfrac{\min\{c_1,\ldots,c_{\alpha}\}^r n^r}{r^r} \times \dfrac{c''^{r'} n^{r'}}{r'^{r'}}  \\
&\geq  \dfrac{{(c \times n)}^{|V(\mathcal{S})|}}{{|V(\mathcal{S})|}^{|V(\mathcal{S})|}}\:, \numberthis \label{AA2}
\end{align*}
where $c= \min\{c'',c_1,\ldots,c_{\alpha}\}$ is a constant, and we have used $r+r'=|V(\mathcal{S})|$.

The proof is then completed by combining (\ref{AA1}) and (\ref{AA2}) with (\ref{E142}), and noting that $d_{v_{\min}} \times n \leq \eta \leq \Delta \times n$, where $d_{v_{\min}}$ is the minimum variable degree in $\lambda(x)$.
}\end{proof}

Theorem~\ref{T1} shows that depending on the relative values of $|V(\mathcal{S})|$ and $|E(\mathcal{S})|$, the expected number of a structure ${\cal S}$ can tend to zero, infinity or a non-zero constant, as $n$ tends to infinity.
The following lemma establishes a connection between the number of cycles in ${\cal S}$, and the value of $|V(\mathcal{S})| - |E(\mathcal{S})|$.

\begin{lem}\label{NEWL1}
Consider a graph ${\cal S}$ with the set of nodes $V(\mathcal{S})$, and the set of edges $E(\mathcal{S})$. If $|V(\mathcal{S})|> |E(\mathcal{S})|$, then $S$ does not contain any cycle. Else, if $|V(\mathcal{S})|< |E(\mathcal{S})|$, then $S$ contains at least two cycles, and if $|V(\mathcal{S})| = |E(\mathcal{S})|$, then $S$ contains only one (simple) cycle.
\end{lem}

Based on Theorem~\ref{T1} and Lemma~\ref{NEWL1}, we have the following corollary.

\begin{cor}\label{R2}
Consider a random ensemble of Tanner graphs and a given subgraph ${\cal S}$ induced by a constant number of variable nodes in such Tanner graphs. Depending on whether $\mathcal{S}$ contains at least two cycles, only one cycle, or no cycle, the average number of structure ${\cal S}$ in the ensemble tends to zero, to a positive constant, or to infinity, as the size of the graphs (length of the codes) tends to infinity.
\end{cor}

In the following subsections, we apply the results of this subsection to different categories of trapping sets, i.e., ETSs, LETSs, ABSs, EABSs and SSs, respectively. 

\subsection{Elementary Trapping Sets (ETS)}
\label{S4}

In this subsection, based on the general result of Theorem~\ref{T1}, we study the asymptotic behavior of ETSs in both regular and irregular LDPC code ensembles.

The following results show that every $(a,b)$ class of ETSs has a consistent behavior in variable-regular LDPC codes, and that the behavior is fully determined by the ratio $b/a$ and the variable degree $d_v$.

\begin{pro}\label{NEWT1}
Consider an $(a,b)$ ETS ${\cal S}$ in a variable-regular Tanner graph with variable degree $d_v$. We then have
\begin{equation}\label{E2}
|V({\cal S})|-|E({\cal S})|=a+\frac{b-ad_v}{2}\:.
\end{equation}
\end{pro}

\begin{proof}
Suppose that ${\cal S}$ is induced in a Tanner graph $G$ by the set of variable nodes $S$. By counting the number of edges in ${\cal S}=G(S)$ from the variable side, we have $|E({\cal S})| = a \times d_v$. We also have
\begin{equation}
|V({\cal S})|=|S|+ |N_o(S)|+ |N_e(S)| = a + b + \frac{ad_v-b}{2} = a + \frac{ad_v+b}{2}\:,
\label{eqwd}
\end{equation}
where the second equality follows from the fact that all the satisfied and unsatisfied check nodes have degree two and one, in the ETS, respectively.
The proof is then completed by subtracting $|E({\cal S})| = a \times d_v$ from (\ref{eqwd}).
\end{proof}

%By Corollary \ref{C1} and (\ref{E2}) we have the following characterization for the class of $(a,b)$ ETSs in variable-regular Tanner graphs.
The next theorem is resulted from Theorem~\ref{T1} and Proposition~\ref{NEWT1}, and describes the asymptotic expected number of ETSs in different $(a,b)$ classes of variable-regular LDPC code ensembles.

\begin{theo} \label{C2}
Consider a random variable-regular LDPC code $C$ with variable degree $d_v$ and length $n$. Denote by $N_{(a,b)}^{ETS}$ the  number of $(a,b)$ elementary trapping sets in $C$. Then, for $a$ being a constant in $n$, in the asymptotic regime where $n \rightarrow \infty$, we have:
$$
\mathbf{E}(N_{(a,b)}^{ETS}) = \Theta(n^{a+\frac{b-ad_v}{2}})\:.
$$
Thus, depending on whether $b/a > d_v-2$, $ b/a < d_v -2$, or $b/a = d_v -2$, the expected value $\mathbf{E}(N_{(a,b)}^{ETS})$ tends to infinity, zero, or a positive constant value in $n$.
\end{theo}

\begin{ex}\label{X1}
Figure \ref{F1} shows three ETS structures in variable-regular Tanner graphs with $d_v=3$, each satisfying one of the conditions in Theorem~\ref{C2}. (Variable nodes, satisfied and unsatisfied
check nodes are shown by full circles, empty squares and full squares, respectively.)
\end{ex}

\begin{figure}[]
\centering
\includegraphics [width=0.43\textwidth]{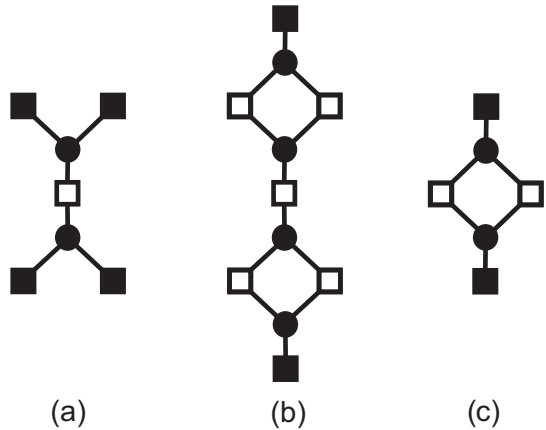}
\caption{Three ETS structures in variable-regular Tanner graphs with $d_v=3$, each satisfying one of the three conditions of Theorem \ref{C2}: (a) A $(2,4)$ ETS  satisfying $b/a > d_v-2$,
         (b) A $(4,2)$ ETS  satisfying $b/a < d_v-2$, and (c) A $(2,2)$ ETS satisfying $b/a = d_v-2$.}
\label{F1}
\end{figure}

In Section~\ref{sec-ets}, we study the case of $b/a = d_v -2$ in more details, and derive closed-form formulas to calculate the 
expected multiplicity of $(a,b)$ ETSs that satisfy $b/a = d_v -2$ in biregular and variable-regular Tanner graphs. 

Unlike variable-regular Tanner graphs, in irregular Tanner graphs, classes of ETSs demonstrate an inconsistent behavior, i.e., in general, in an $(a,b)$ class, one can find at least two structures whose expected numbers tend to different values (infinity, zero or a non-zero constant), as $n \rightarrow \infty$. This is explained in the following example for the $(4,2)$ class.

\begin{ex}
Figure \ref{F2} shows three ETS structures, all in the $(4,2)$ class, but each with a different asymptotic behavior. While the asymptotic expected values of the leftmost and the rightmost structures are infinity and zero, respectively, for the middle structure, the asymptotic expected value is a positive constant (see, Corollary~\ref{R2}).
\end{ex}

 \begin{figure}[]
 \centering
 \includegraphics [width=0.5\textwidth]{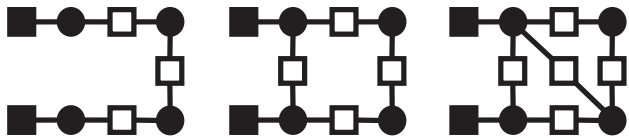}
  \caption{Three $(4,2)$ ETS structures in irregular Tanner graphs.}
 \label{F2}
 \end{figure}

In Section~\ref{sec-ets}, we further discuss the ETSs whose average multiplicity tends to a constant in irregular graphs asymptotically and derive formulas to compute such expected values. 

\subsection{Leafless Elementary Trapping Sets (LETS)}
\label{SS4-2}

A LETS is a special case of an ETS, and thus the general results presented in Proposition~\ref{NEWT1} and Theorem~\ref{C2} are also applicable to LETS structures of random variable-regular Tanner graphs. For LETSs, however,
from the three scenarios of $b/a > d_v-2$, $b/a < d_v-2$, and $b/a = d_v-2$, only the last two can happen. This is proved in the following lemma.

\begin{lem}\label{L1}
For any LETS structure in the $(a,b)$ class of a variable-regular LDPC code with variable degree $d_v$, we have $b/a \leq d_v-2$.
\end{lem}

\begin{proof}
Consider an $(a,b)$ LETS structure induced by the set of variable nodes $S$. Counting the number of edges in the subgraph $G(S)$ from the two sides of the graph, we have $a \times d_v = b + 2|N_e(S)|$.
Since $G(S)$ is a LETS structure, each variable node is connected to at least 2 satisfied check nodes. This implies $2 |N_e(S)| \geq 2a$, which together with the previous equation complete the proof.
\end{proof}

\begin{theo} \label{C3}
Consider a random variable-regular LDPC code $C$ with variable degree $d_v$ and length $n$. Denote by $N_{(a,b)}^{LETS}$ the  number of $(a,b)$ leafless elementary trapping sets in $C$.
Then, for $a$ being a constant in $n$, in the asymptotic regime where $n \rightarrow \infty$, we have:
$$
\mathbf{E}(N_{(a,b)}^{LETS}) = \Theta(n^{a+\frac{b-ad_v}{2}})\:.
$$
Thus, depending on whether $ b/a < d_v -2$, or $b/a = d_v -2$, the expected value $\mathbf{E}(N_{(a,b)}^{LETS})$ tends to zero, or a positive constant value in $n$.
\end{theo}

Based on Theorem~\ref{C3}, it is clear that, in the asymptotic regime, only the LETS classes with $b/a=d_v-2$ are non-empty.
All the structures in such classes contain only one cycle, and we thus have the following result.

\begin{cor}\label{NEWL2}
In random variable-regular LDPC codes with variable degree $d_v$ and length $n$, as $n \rightarrow \infty$, the only non-empty classes of local $(a,b)$ LETS structures are those with $b/a=d_v-2$.
For each such class, all structures within the class correspond to a simple (chordless) cycle of length 2a.
%For each such class, each structure within the class corresponds to a simple (chordless) cycle of length $2a$ (for a given structure it is enough to remove all nodes of degree one, to obtain a simple cycle, and vice versa).
\end{cor}

For irregular LDPC codes, the following result is in parallel with Lemma~\ref{L1}.

\begin{lem}\label{L3}
For any LETS structure ${\cal S}$ in an irregular LDPC code, we have $|V({\cal S})| \leq |E({\cal S})|$.
\end{lem}

\begin{proof}
To prove the lemma, we show that any LETS structure ${\cal S}$ has a cycle. In ${\cal S}$, each variable node is connected to at least two satisfied check nodes. On the other hand, the degree of each satisfied check node is two. Now, if we consider the subgraph formed by variable nodes, satisfied check nodes and the edges in between, we have a graph with minimum degree two. It is well-known that a graph with the minimum degree two has at least one cycle.
\end{proof}

Based on Lemma~\ref{L3}, for irregular LDPC codes also, the asymptotic expected number of LETS structures is either zero or a constant non-zero value. The classes of LETS structures in irregular codes also demonstrate an inconsistent behavior,
and among all the structures within an $(a,b)$ class, only those that correspond to simple cycles of length $2a$ have an asymptotically non-zero expected value for their multiplicity. We thus have the following result.

\begin{pro}\label{NT1}
Consider random irregular LDPC codes with $d_{v_{min}} \geq 2$, and length $n$. As $n \rightarrow \infty$, for a constant value of $a$, the sum of the expected number of LETS structures in all the $(a,b)$ classes, for different values of $b$, tends to the expected value of the number of simple cycles of length $2a$ in the code.
\end{pro}

The cycle structure of random regular and irregular Tanner graphs was studied in~\cite{dehghan2016new}. In particular, it was shown in~\cite{dehghan2016new} that the asymptotic
expected value of the number of cycles of length $c$ in irregular graphs with variable degrees $\{d_i\}_{i=1}^n$ and check node degrees $\{d'_i\}_{i=1}^{n'}$, as the size of the graph tends to infinity,
is approximated by
\begin{equation}
\mathbf{E}(N_c) \approx \dfrac{\displaystyle\Big((\frac{2}{|E|}\displaystyle\sum_{i=1}^{n}{{d_i}\choose {2}}) (\frac{2}{|E|}\displaystyle\sum_{i=1}^{n'}{{d'_i}\choose {2}})\Big)^{c/2}}{c}\:,
\label{eq12}
\end{equation}
where $|E| =\sum_{i=1}^n d_i =\sum_{i=1}^{n'} d'_i$ is the number of edges in the graph. This result for variable-regular graphs with variable degree $d_v$ reduces to
\begin{equation}
\mathbf{E}(N_c) \approx \dfrac{\displaystyle\Big( (d_v-1)\displaystyle(\frac{2}{|E|}\displaystyle\sum_{i=1}^{n'}{{d'_i}\choose {2}})\Big)^{c/2}}{c}\:,
\label{eq23}
\end{equation}
and for $(d_v,d_c)$ biregular graphs to 
\begin{equation}
\mathbf{E}(N_c) \sim \dfrac{\Big((d_v-1)(d_c-1)\Big)^{c/2}}{c}\:.
\label{eqsd}
\end{equation}

Considering that in the asymptotic regime of $n \rightarrow \infty$, by Corollary \ref{R2}, the expected number of cycles with chords tends to zero, one can use the above approximations for chordless cycles, and use them along with Corollary~\ref{NEWL2}, to obtain asymptotic estimates on the average number of LETS structures in different classes of regular graphs.

For irregular graphs, the result of Proposition~\ref{NT1} together with the approximation (\ref{eq12}) can be used to estimate the sum of $(a,b)$ LETS structures for a given $a$ and different $b$ values. In Section~\ref{LETS-IRREG}, we fine tune our analysis for irregular graphs and compute the expected number of $(a,b)$ LETS structures for any given values of $a$ and $b$ for which the expected multiplicity tends to a constant.
 
\subsection{Absorbing Sets (ABS)}
\label{ABS}

%Absorbing sets are known to be the fixed points of hard-decision algorithms such as Gallager A~\cite{},~\cite{}. They have also been identified as the harmful structures for quantized decoders~\cite{}
The following result relates the number of nodes and edges of an ABS.

\begin{lem}
For any ABS structure ${\cal S}$ within a Tanner graph with $d_{v_{\min}} \geq 2$, we have $|V({\cal S})| \leq |E({\cal S})|$.
\label{lemxf}
\end{lem}

\begin{proof}
Let ${\cal S} = G(S)$ be the ABS induced in the Tanner graph $G$ by the set of variable nodes $S$. By the definition of an ABS, each variable node in $S$ is connected to more nodes in $N_e(S)$ than in $N_o(S)$.
Each node in $S$ has thus at least two neighbors in  $N_e(S)$. On the other hand, the degree of each node in $N_e(S)$ within $G(S)$ is at least two. Thus, if we consider the subgraph of $G(S)$ containing the
nodes in $S$ and $N_e(S)$, and the edges in between, we obtain a graph with minimum degree two. Such a graph, thus, has a cycle. This completes the proof.
\end{proof}

Based on Theorem~\ref{T1} and Lemma~\ref{lemxf}, it is clear that the average number of any ABS structure tends to either zero or a positive constant (not to infinity) as the block length tends to infinity. The following theorem distinguishes between the two cases depending on the variable degrees of the Tanner graph.

\begin{theo}
Consider random Tanner graphs with variable node degree distribution $\lambda(x)$. If $d_{v_{\min}} \geq 4$, then all the classes of local ABSs have zero multiplicity asymptotically. Otherwise,
if $d_{v_{\min}} = 3$,  then all the local $(a,b)$ classes of ABSs with $a \geq 2$ and $b \neq a$, have zero multiplicity asymptotically. In this case ($d_{v_{\min}} = 3$), the structures in the local $(a,a)$ class with asymptotically non-zero multiplicity
all correspond to simple cycles of length $2a$ consisting only of degree-3 variable nodes.
Finally, if $d_{v_{\min}} = 2$, then only the local ABS structures whose variable degrees are only $2$ or $3$ can have non-zero multiplicity asymptotically. In the $(a,b)$ classes
with a given $a$ and different values of $b \leq a$, such structures are all simple cycles of length $2a$.
\label{proABS}
\end{theo}

\begin{proof}
If $d_{v_{\min}} \geq 4$, then each variable node in the ABS must be connected to at least three satisfied check nodes. Now, consider the subgraph that consists of the variable nodes of the ABS, its satisfied check nodes and the edges in between. This subgraph has minimum degree two, and a node of degree at least three, and thus contains at least two cycles. This means that the ABS itself contains at least two cycles and thus, based on Corollary~\ref{R2}, its multiplicity is zero asymptotically. The same argument applies to any ABS structure that has at least one node with degree greater than or equal to four (for the cases with $d_{v_{\min}} = 3$ or $2$).

If $d_{v_{\min}} = 3$, thus, all ABS structures have asymptotically zero multiplicity, except those whose variable nodes, all, have degree three. In this case, by the definition of an ABS, we must have $b \leq a$, as each degree-3 variable node must be connected to at least two satisfied check nodes. Now consider the case where $a \geq 2$ and $b < a$. For this case, we consider two scenarios and show that for both scenarios the ABS structures will have more than one cycle and thus their multiplicity is zero asymptotically: (1) all unsatisfied check nodes have degree one, (2) at least one unsatisfied check node $c$ has degree $3$ or larger. In the first scenario, since $b < a$, there must exist a variable node $v$ in the ABS that has no connection to unsatisfied check nodes. Node $v$ is thus connected to three satisfied check nodes, and with the same argument presented before, the ABS will have more than one cycle. In the second scenario, consider the subgraph of the ABS consisting of variable nodes, satisfied check nodes and the unsatisfied check node $c$ plus all the edges in between. This subgraph has minimum degree $2$ and a node with degree at least $3$, and thus has more than one cycle. Therefore, for the case of $d_{v_{\min}} = 3$, the only classes of ABSs with asymptotically non-zero multiplicity are $(a,a)$ classes. The only elements within the $(a,a)$ class whose multiplicity is non-zero asymptotically are simple cycles consisting of only degree-3 variable nodes, where each variable node is connected to one unsatisfied check node of degree one, and the degree of all satisfied check nodes are two.

For graphs with $d_{v_{\min}} = 2$, the proof of the statement of the proposition is similar. In this case, the only structures whose average multiplicity tends to a constant are those with variable degrees only $2$ or $3$, that contain only a simple cycle.
\end{proof}

 \begin{figure}[]
 \centering
 \includegraphics [width=0.5\textwidth]{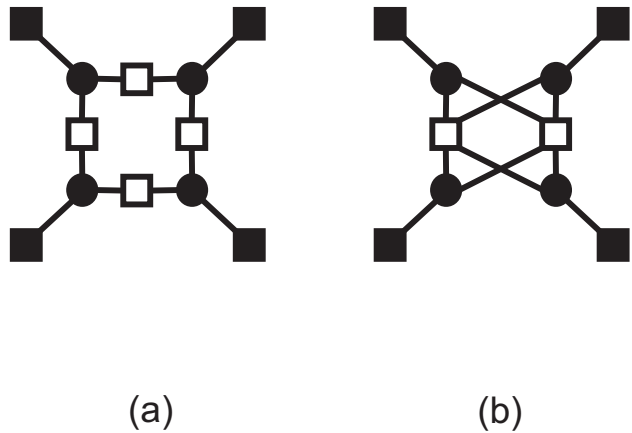}
  \caption{Two $(4,4)$ ABS structures with different asymptotic behaviors in Tanner graphs with $d_{v_{min}} = 3$.}
 \label{figxx}
 \end{figure}

From Theorem~\ref{proABS}, it can be seen that any class of ABSs has a consistent behavior in any ensemble of variable-regular or irregular LDPC codes with $d_{v_{\min}} \geq 4$.
The same applies to any class of $(a,b)$ ABSs with $a \neq b$ in any ensemble of LDPC codes with $d_{v_{\min}} = 3$. One can, however, easily provide examples where two ABS structures within the same
$(a,a)$ class have different asymptotic behavior in an ensemble with $d_{v_{\min}} = 3$. This is demonstrated in Fig.~\ref{figxx} for the $(4,4)$ class. While the structure in Fig.~\ref{figxx}(a) has only one cycle and thus has an average multiplicity tending to a non-zero constant, the structure in Fig.~\ref{figxx}(b) contains more than one cycle and thus its multiplicity is zero asymptotically.  For ensembles with $d_{v_{\min}} = 2$, it can be seen that, in general, an $(a,b)$ class of ABSs
with $a \geq 2$ and $b \leq a$, has an inconsistent behavior (see Fig.~\ref{figgg}, for two structures with different asymptotic behaviors in the $(4,1)$ class).

 \begin{figure}[]
 \centering
 \includegraphics [width=0.5\textwidth]{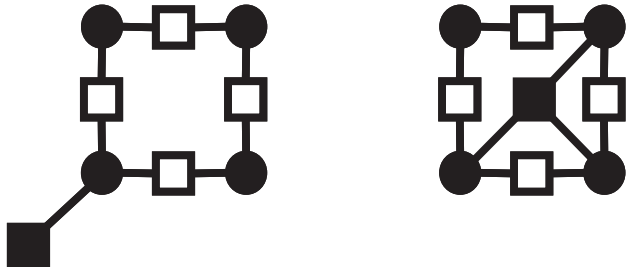}
  \caption{Two $(4,1)$ ABS structures with different asymptotic behaviors in Tanner graphs with $d_{v_{min}} = 2$.}
 \label{figgg}
 \end{figure}

In addition to the asymptotic results of Theorem~\ref{proABS}, the following finite-length result applies to ABSs of variable-regular LDPC codes.

\begin{lem}\label{Ob1}
There is no $(a,b)$ ABS structure with $b > a(\lceil \frac{d_v}{2} \rceil - 1)$, in a variable-regular Tanner graph with variable degree $d_v$.
\end{lem}

\begin{proof}
By the definition of an ABS, each variable node is connected to at least $\lfloor \frac{d_v}{2} \rfloor +1$ satisfied check nodes. The number of edges in the induced subgraph connected to unsatisfied check nodes is thus at most 
$a \times d_v - a (\lfloor \frac{d_v}{2} \rfloor +1)$ or $a (\lceil \frac{d_v}{2} \rceil -1)$.
%We note that for an EABS ${\cal S}$, the condition $b/a=d_v - 2$, is equivalent to $|V({\cal S})|=|E({\cal S})|$, which in turn means that ${\cal S}$ contains only a simple cycle. The condition $b/a=d_v - 2$ thus implies that every variable node in ${\cal S}$ is connected to two satisfied and $d_v-2$ unsatisfied check nodes. By the definition of an ABS, however, each variable node in ${\cal S}$ must be connected to more satisfied than unsatisfied check nodes. This implies that $d_v-2 < 2$ or $d_v < 4$.
\end{proof}

\begin{cor}
There is no $(a,b)$ ABS structure with $b/a=d_v - 2$, in a variable-regular Tanner graph with variable degree $d_v \geq 4$.
\end{cor}

In Section~\ref{LETS-IRREG}, we complement the results of Theorem~\ref{proABS}, by calculating the asymptotic expected multiplicity of ABSs for irregular codes with $d_{v_{\min}}=3$ and $d_{v_{\min}}=2$, in $(a,b)$ classes with $b \leq a$ (see Corollaries~\ref{cor329} and \ref{cor330}). 

\subsection{Elementary Absorbing Sets (EABS)}
\label{SS4-3}

Elementary ABSs are a special case of ABSs. All the results presented in the previous subsection are thus applicable to EABSs as well. On the other hand, EABSs, for codes with $d_{v_{\min}} \geq 2$, are a special case of LETSs.
Therefore, the results presented in Subsection~\ref{SS4-2} are also applicable to EABSs. In fact, for the variable-regular graphs with $d_v=2$ or $d_v=3$, or irregular
graphs with variable degrees only $2$ and $3$, the sets of EABSs and LETSs are identical.

In the asymptotic regime of $n \rightarrow \infty$, for variable-regular Tanner graphs with $d_v \geq 4$, using Theorem~\ref{proABS},
one can see that all $(a,b)$ EABS classes are empty. For the cases of $d_v=2$ and $d_v=3$, based of Corollary~\ref{NEWL2}, the only non-empty classes (asymptotically) are those
with $b/a = d_v-2$, whose members are simple cycles of length $2a$. The average multiplicity of such classes can then be approximated by (\ref{eq23}) or (\ref{eqsd}).

In general, for an irregular Tanner graph, the following result shows that among ABS structures, only those that are elementary can possibly have non--zero multiplicity in the asymptotic regime of $n \rightarrow \infty$. This implies that for
graphs with $d_{v_{\min}} \geq 2$, the asymptotic results presented in Theorem~\ref{proABS}, related to ABS structures with constant average multiplicity, applies directly to EABSs.

\begin{pro}
Any non-elementary local ABS structure ${\cal S}$ in a Tanner graph with $d_{v_{\min}} \geq 2$ contains more than one cycle. The multiplicity of ${\cal S}$ in a random Tanner graph thus tends to zero as the size of the graph tends to infinity.
\label{profd}
\end{pro}

\begin{proof}
The non-elementary ABS structure ${\cal S}$ has either ($a$) a satisfied check node with degree $4$ or larger, or ($b$) an unsatisfied check node of degree $3$ or larger. In addition, each variable node in ${\cal S}$ is connected to at least two satisfied check nodes. For Case ($a$), consider the subgraph of ${\cal S}$ consisting of all the variable nodes in ${\cal S}$, all the satisfied check nodes and the edges in between. This subgraph has minimum degree $2$ and has a node with degree $4$ or larger. It thus contains more than one cycle, and so does the ABS structure ${\cal S}$, itself. For Case ($b$), consider the subgraph of ${\cal S}$ containing all the variable nodes, all the satisfied check nodes and the unsatisfied check node with degree $3$ or larger, as well as all the edges in between. This subgraph has minimum degree $2$ and has a node with degree $3$ or larger, and thus contains more than one cycle.
\end{proof}

\subsection{Stopping sets (SS)}
\label{SS4-5}

In stopping sets, each check node is connected to at least two variable nodes. For Tanner graphs with $d_{v_{\min}} \geq 2$, therefore, any SS structure has a minimum degree at least $2$, and as a result, contains at least one cycle.
If we, however, further limit the degree distribution of Tanner graphs to $d_{v_{\min}} \geq 3$, any SS structure will have a minimum degree of two and at least one node with degree $3$ or larger. This implies that the SS structure will contain more than one cycle.

\begin{lem}\label{NEWL6}
Any SS structure in a Tanner graph with $d_{v_{\min}} \geq 2$ contains at least one cycle, i.e., $|V({\cal S})|\leq|E({\cal S})|$. For Tanner graphs with $d_{v_{\min}} \geq 3$, however,
any SS structure contains more than one cycle, i.e.,  $|V({\cal S})| < |E({\cal S})|$.
\end{lem}

It is easy to see that, in general, the only SS structures that contain only one cycle are those with all variable nodes having degree $2$. For such stopping sets, we have the following result.

\begin{lem}
Consider a stopping set that belongs to the $(a,b)$ class, and in which all the variable nodes have degree $2$. We then have $b < a$.  Moreover, among all such SS structures, only those with $b=0$
can contain only one cycle. Those in other classes all have more than one cycle.
\label{lem56}
\end{lem}

\begin{proof}
Since in an SS, each check node has degree at least $2$, the number of check nodes in the SS (with all variable nodes having degree $2$)  must be less than or equal to the number of variable nodes and thus $b \leq a$. If $b=a$, however, since the degree of an odd-degree check node in the SS is at least $3$, then there will be $3b$ edges connected to odd-degree check nodes. This, alone, will be larger than the total number of edges connected to variable nodes. Thus, we must have $b < a$. For the second part of the lemma, we notice that any SS with $b > 0$ has at least one check node of degree $3$ or larger, and thus contains more than one cycle.
\end{proof}

Based on Lemma~\ref{lem56}, the only SS structures with only one cycle are codewords ($b=0$). In such structures, the number of variable and check nodes is the same, and all variable nodes and check nodes have degree $2$. In fact, these structures are all simple cycles of length $2a$ formed by degree-2 variable nodes.

\begin{theo}
Any $(a,b)$ class of local stopping sets has a consistent behavior in any ensemble of LDPC codes with $d_{v_{\min}} \geq 3$, with the asymptotic multiplicity equal to zero.
Moreover, the same result applies to LDPC codes with $d_{v_{\min}} = 2$ and stopping set classes with $b > 0$.  In ensembles with $d_{v_{\min}} = 2$, $(a,0)$ classes, in general,
however, do not demonstrate a consistent behavior. The only stopping set structures in these classes whose average multiplicity tends to a non-zero constant by increasing the block length are simple cycles of length $2a$ consisting of
$a$ degree-$2$ variable nodes. All the remaining structures have an asymptotic multiplicity of zero.
\label{thm54}
\end{theo}

Fig.~\ref{figee} shows two stopping sets, both having only variable nodes with degree $2$, and both in the $(4,0)$ class, but with different asymptotic behavior.

In the following section, to complement the result of Theorem~\ref{thm54}, in Corollary~\ref{cor543}, we compute the asymptotic expected multiplicity of stopping sets in the $(a,0)$ class.
  
 \begin{figure}[]
 \centering
 \includegraphics [width=0.5\textwidth]{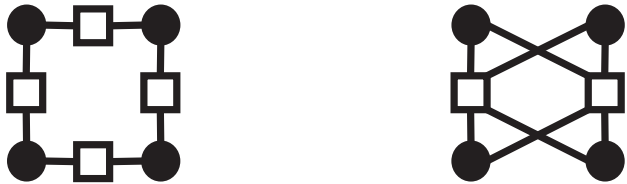}
  \caption{Two $(4,0)$ SS structures, consisting of only degree-$2$ variable nodes, with different asymptotic behavior.}
 \label{figee}
 \end{figure}

\section{Asymptotic Expected Multiplicity of Stopping Set, ABS and LETS Structures in Irregular Bipartite Graphs}
\label{LETS-IRREG}

The key results that enable us to find the asymptotic expected value of stopping set, ABS and LETS structures in irregular graphs are the following theorem, and its generalization in Theorem~\ref{Th57}, that compute the 
expected number of cycles of a given length with a given combination of variable and check node degrees.

\begin{theo}\label{Th45}
%Let $\lambda(x) =\sum_{i=2}^{z} \lambda_i x^{i-1}$ and $\rho (x) = \sum_{i=2}^{w} \rho_i x^{i-1}$ be degree distributions where the coefficients $\lambda_i$ and $\rho_i$ represent the fraction of edges connected to variable and check nodes of degree $i$, respectively.
Consider the probability space $\mathcal{G}$ of all bipartite graphs with variable and check node degree distributions $\lambda(x)= \sum_{i=z'}^{z} \lambda_i x^{i-1}$ and $\rho (x)= \sum_{i=w'}^{w} \rho_i x^{i-1}$, respectively,
where the graphs in $\mathcal{G}$ are selected uniformly at random.  
Also, assume that $z'=d_{v_{\min}} \geq 2$, $w'= d_{c_{\min}} \geq 2$, and that $z=d_{v_{\max}}$ and $w=d_{c_{\max}}$ are fixed natural numbers, and let the number of variable nodes $n$ and the number of check nodes $n'$ tend to infinity.
%Suppose that the graphs in $\mathcal{G}$ are selected uniformly at random.
Let $\vec{\alpha}_c=(\alpha_{z'},\ldots, \alpha_{z})$ and  $\vec{\beta}_c=(\beta_{w'},\ldots, \beta_{w} )$ be sequences of non-negative integers such that $\sum_{i=z'}^{z} \alpha_i=c/2$ and $\sum_{j=w'}^{w} \beta_j=c/2$.
For given vectors $\vec{\alpha}_c$, $\vec{\beta}_c$ and  $G \in \mathcal{G}$, denote by $N_{\vec{\alpha}_c,\vec{\beta}_c}(G)$ the number of cycles of length $c$ in $G$ such that each 
cycle has $\alpha_i$ variable nodes with degree $i$, for each $i=z',\ldots, z$, and $\beta_j$ check nodes with degree $j$, for each $j=w',\ldots, w$. 
Then, for any fixed even value of $c \geq 4$, we have %the random variables $N_{\vec{\alpha}_c,\vec{\beta}_c}$, has the expected value
\begin{equation}
\mathbf{E}(N_{\vec{\alpha}_c,\vec{\beta}_c})\sim \dfrac{\displaystyle {{c/2}\choose{\alpha_{z'},\ldots, \alpha_{z}}}\prod_{i=z'}^{z}(\lambda_i(i-1))^{\alpha_i} {{c/2}\choose{\beta_{w'},\ldots, \beta_{w}}}\prod_{j=w'}^{w}(\rho_j(j-1))^{\beta_j} }{c}\:.
\label{eq09}
\end{equation}
\end{theo}

\begin{proof}{
%In a random bipartite graph, for each node $z$, we consider a bin that contains $deg (z)$ cells. We now consider a random perfect matching to pair the cells on the $U$ side of the graph to the cells on the $W$ side. Corresponding to each matching, we construct a  bipartite graph such that if there is an edge between two cells, then we place an edge between the corresponding nodes (bins) in the  bipartite graph. The bipartite graphs are thus represented as images of the so-called {\em configurations} that are obtained from the random perfect matchings. Clearly, there are $N(|E(G)|)=|E(G)|!$ configurations, where $|E(G)|$ is the number of edges in the graph. Also, note that  given a set of $\ell$ fixed edges, there are $N_{\ell}(|E(G)|)=(|E(G)|-\ell)!$ configurations containing those edges.

%Since each bipartite graph corresponds to precisely $$(d_1!)(d_2!)\ldots (d_n!)(d_1'!)( d_2'!)\ldots( d_{m}'!) $$ matchings, a bipartite graph can be chosen uniformly at random by choosing a matching uniformly at random and rejecting the result if it has loops or multiple edges. Thus, we can consider a probability space where the random bipartite graphs are selected uniformly at random.

%It is clear that in our model there is no configuration with loop. Given a set of $\ell$ fixed edges, let $M_{\ell}(|E(G)|)$ be the number configurations containing those edges and does not have any parallel edges. As it was shown in Corollary 2.18 in \cite{MR1864966}, we have: $ M_{\ell}(|E(G)|)/ M_{0}(|E(G)|)\sim N_{\ell}(|E(G)|)/ N(|E(G)|)= (|E(G)|-\ell)!/(|E(G)|)!$.

We use the same model and notations as described in the proof of Theorem~\ref{T1}, where bipartite graphs are images of configurations. Rather than working in the probability space $\mathcal{G}$ to find the expected value, we then work in the probability space corresponding to the configurations. Using an approach similar to the one used in~\cite{dehghan2016new}, one can show that the expected values in the two spaces are asymptotically equal. 

For a configuration, we define an ($\vec{\alpha}_c,\vec{\beta}_c$)-cycle  to be a set of $c$ edges, like $\{e_1, e_2, \ldots , e_c\}$, connecting $c$
distinct bins, like $D_{i_1}, \ldots, D_{i_c}$. The connections are 
such that for each $j \in \{1,\ldots,c\}$, the edge $e_j$, connects a cell in bin $D_{i_j}$ to a cell in bin $D_{i_{j+1}}$, where $D_{i_{c+1}}=D_{i_{1}}$,
and the two cells in each bin $D_{i_{j}}$, connected to the two edges $e_{j}$ and $e_{j-1}$, are distinct ($e_0 = e_c$).
Also, for the bins corresponding to variable nodes, for each $i$, there are $\alpha_i$ bins, each with $i$ cells, %(i.e. a bin of size $i$ is a bin that contains $i$ cells) 
and for the bins corresponding to check nodes, for each $j$, there are $\beta_j$ bins, each with $j$ cells.
We now compute the number of ($\vec{\alpha}_c,\vec{\beta}_c$)-cycles, ${\cal C}_{\vec{\alpha}_c,\vec{\beta}_c}$, in a configuration.

%We call these two cells, the {\it first cell} and the {\it second cell}.
%Without loss of generality 
For a bin with $i$ cells, %Suppose that bin $i$ contains $d_i$ cells. We thus 
we have $i(i-1)$ choices for the two cells that are at the end points of the edges connected to the bin. Hence, in order to choose all the cells on both sides of the graph, we have
$$\Big( \displaystyle\sum_{\sigma_{z'} \subset U_{z'}, |\sigma_{z'}|=\alpha_{z'} \atop  {\vdots \atop \sigma_z \subset U_z, |\sigma_z|=\alpha_z } }\prod_{i=z'}^{z} (i(i-1))^{\alpha_i} \Big)\Big( \displaystyle\sum_{\sigma_{w'} \subset W_{w'}, |\sigma_{w'}|=\beta_{w'} \atop  {\vdots \atop \sigma_w \subset W_w, |\sigma_w|=\beta_w } }\prod_{j=w'}^{w} (j(j-1))^{\beta_j} \Big)$$
choices, where $U_i$ and $W_i$ are the sets of variable and check nodes with degree $i$, respectively. To count the number of ($\vec{\alpha}_c,\vec{\beta}_c$)-cycles in a configuration, 
we also need to consider different orderings of the $c/2$ bins on each side of the graph.
This results in
\begin{equation} \label{E1}
 {\cal C}_{\vec{\alpha}_c,\vec{\beta}_c} =\Big( \displaystyle\sum_{\sigma_{z'} \subset U_{z'}, |\sigma_{z'}|=\alpha_{z'} \atop  {\vdots \atop \sigma_z \subset U_z, |\sigma_z|=\alpha_z } }\prod_{i=z'}^{z} (i(i-1))^{\alpha_i} \Big)\Big( \displaystyle\sum_{\sigma_{w'} \subset W_{w'}, |\sigma_{w'}|=\beta_{w'} \atop  {\vdots \atop \sigma_w \subset W_w, |\sigma_w|=\beta_w } }\prod_{j=w'}^{w} (j(j-1))^{\beta_j} \Big)\Big( \dfrac{(\frac{c}{2})!(\frac{c}{2})!}{c} \Big)\:,
\end{equation}
where the division by $c$ is for counting each cycle in the above process $c$ times.
%We note that given a set of $\ell$ fixed edges, there are $(|E(G)|-\ell)!$ configurations containing those edges.
%The image of a ($\vec{\alpha}_c,\vec{\beta}_c$)-cycle in a configuration is a ($\vec{\alpha}_c,\vec{\beta}_c$)-cycle in the graph.
We thus have

\begin{align*}
\mathbf{E}(N_{\vec{\alpha}_c,\vec{\beta}_c}) & \sim \dfrac{{\cal C}_{\vec{\alpha}_c,\vec{\beta}_c}\times(|E(G)|-c)!}{|E(G)|!}\\
%       &= \Big( \displaystyle\sum_{\sigma_2 \subset U_2, |\sigma_2|=\alpha_2 \atop  {\vdots \atop \sigma_z \subset U_z, |\sigma_z|=\alpha_z } }\prod_{i=2}^{z} (i(i-1))^{\alpha_i} \Big)\Big( \displaystyle\sum_{\sigma_2 \subset W_2, |\sigma_2|=\beta_2 \atop  {\vdots \atop \sigma_w \subset W_w, |\sigma_w|=\beta_w } }\prod_{j=2}^{w} (j(j-1))^{\beta_j} \Big)\Big( \dfrac{(\frac{c}{2})!(\frac{c}{2})!}{c} \Big)  \dfrac{(|E(G)|-c)!}{|E(G)|!}\\
       &\sim \Big( \displaystyle\sum_{\sigma_{z'} \subset U_{z'}, |\sigma_{z'}|=\alpha_{z'} \atop  {\vdots \atop \sigma_z \subset U_z, |\sigma_z|=\alpha_z } }\prod_{i=z'}^{z} (i(i-1))^{\alpha_i} \Big)\Big( \displaystyle\sum_{\sigma_{w'} \subset W_{w'}, |\sigma_{w'}|=\beta_{w'} \atop  {\vdots \atop \sigma_w \subset W_w, |\sigma_w|=\beta_w } }\prod_{j=w'}^{w} (j(j-1))^{\beta_j} \Big)\Big( \dfrac{(\frac{c}{2})!(\frac{c}{2})!}{c} \Big) \dfrac{1}{|E(G)|^c}.  \numberthis \label{E222}
\end{align*}

In the following, we simplify (\ref{E222}) in a number of steps. First,
\begin{equation}\label{EN2}
\Big( \displaystyle\sum_{\sigma_{z'} \subset U_{z'}, |\sigma_{z'}|=\alpha_{z'} \atop  {\vdots \atop \sigma_z \subset U_z, |\sigma_z|=\alpha_z } }\prod_{i=z'}^{z} (i(i-1))^{\alpha_i} \Big)
        = \prod_{i=z'}^{z} {{|U_i|}\choose {\alpha_i}}\times\prod_{i=z'}^{z}(i(i-1))^{\alpha_i}\:.
\end{equation}
Also, for any two integers  $k, s$, where $k$ tends to infinity and $s$ is a constant number, we have
\begin{equation}\label{EN10}
\displaystyle{k \choose s }\sim\dfrac{k^s}{s!}\:.
\end{equation}
By (\ref{EN2}) and (\ref{EN10}), we obtain
\begin{equation}\label{E3}
\Big( \displaystyle\sum_{\sigma_{z'} \subset U_{z'}, |\sigma_{z'}|=\alpha_{z'} \atop  {\vdots \atop \sigma_z \subset U_z, |\sigma_z|=\alpha_z } }\prod_{i=z'}^{z} (i(i-1))^{\alpha_i} \Big) \sim
 \prod_{i=z'}^{z} \dfrac{|U_i|^{\alpha_i}}{(\alpha_i)!}(i(i-1))^{\alpha_i}\:.
\end{equation}
Moreover, by the definition of degree distributions, % $\lambda(x) =\sum_{i=2}^{z} \lambda_i x^{i-1}$ and $\rho (x) = \sum_{i=2}^{w} \rho_i x^{i-1}$, 
for each $i$, we have:
\begin{equation}\label{EN1}
\dfrac{|U_i|}{n} = \dfrac{|U_i|}{\sum_{j}|U_j|}= \dfrac{\lambda_i/i}{\sum_{j}\lambda_j/j}\, , \,\,\,\,\,\, \dfrac{|W_i|}{n'} = \dfrac{|W_i|}{\sum_{j}|W_j|}= \dfrac{\rho_i/i}{\sum_{j}\rho_j/j}\:.
\end{equation}

Consequently, by (\ref{EN1}) and (\ref{E3}), we have
\begin{equation}\label{E4}
\Big( \displaystyle\sum_{\sigma_{z'} \subset U_{z'}, |\sigma_{z'}|=\alpha_{z'} \atop  {\vdots \atop \sigma_z \subset U_z, |\sigma_z|=\alpha_z } }\prod_{i=z'}^{z} (i(i-1))^{\alpha_i} \Big) \sim \prod_{i=z'}^{z} \Big [\Big (\dfrac{n}{\sum_{j}\lambda_j/j}\Big )^{{\alpha_i}} \dfrac{(\lambda_i(i-1))^{\alpha_i}}{ (\alpha_i)!}\Big ]\:.
\end{equation}
Using ${{c/2}\choose{\alpha_{z'},\ldots, \alpha_{z}}}=\frac{(c/2)!}{(\alpha_{z'})!\cdots (\alpha_{z})!} $ in (\ref{E4}), we further obtain
\begin{align*}
\Big( \displaystyle\sum_{\sigma_{z'} \subset U_{z'}, |\sigma_{z'}|=\alpha_{z'} \atop  {\vdots \atop \sigma_z \subset U_z, |\sigma_z|=\alpha_z } }\prod_{i=z'}^{z} (i(i-1))^{\alpha_i} \Big) &\sim\displaystyle \dfrac{ {{c/2}\choose{\alpha_{z'},\ldots, \alpha_{z}}}}{(c/2)!}\prod_{i=z'}^{z} \Big [\Big (\dfrac{n}{\sum_{j}\lambda_j/j}\Big )^{{\alpha_i}} (\lambda_i(i-1))^{\alpha_i}\Big ]\\
&= \displaystyle \dfrac{ {{c/2}\choose{\alpha_{z'},\ldots, \alpha_{z}}}}{(c/2)!}\Big(\dfrac{n}{\sum_{j}\lambda_j/j}\Big )^{{c/2}}\prod_{i=z'}^{z} (\lambda_i(i-1))^{\alpha_i}\:. \numberthis \label{EN3}
\end{align*}

By counting the number of edges from the variable node side of the graph, and by (\ref{EN1}), we have
\begin{align*}
|E(G)|^{c/2}&=\Big ( \sum_{i=z'}^{z} i|U_i|  \Big)^{c/2}\\
         &=\Big ( \sum_{i=z'}^{z} \dfrac{n\lambda_i}{\sum_{j}\lambda_j/j}  \Big)^{c/2}\\
 %        &=\Big ( \dfrac{n}{\sum_{j}\lambda_j/j} \sum_{i=2}^{z}  \lambda_i \Big)^{c/2}\\
         &=\Big ( \dfrac{n}{\sum_{j}\lambda_j/j} \Big)^{c/2}\:.  \numberthis \label{EN4}
\end{align*}

By substituting (\ref{EN3}) and (\ref{EN4}), for both sides of the graph, in (\ref{E222}), we obtain (\ref{eq09}).
%\begin{align*}
%\mathbf{E}(N_{\vec{\alpha}_c,\vec{\beta}_c}) & \sim \dfrac{\displaystyle {{c/2}\choose{\alpha_{z'},\ldots, \alpha_{z}}}\prod_{i=z'}^{z}(\lambda_i(i-1))^{\alpha_i} {{c/2}\choose{\beta_{w'},\ldots, \beta_{w}}}\prod_{j=w'}^{w}(\rho_j(j-1))^{\beta_j} }{c}\:.\numberthis \label{eq09}
%\end{align*}

}\end{proof}  

\begin{rem}\label{R46}
For a $(d_v,d_c)$-regular Tanner graph, the result of Theorem~\ref{Th45}, obtained by replacing $\alpha_{d_v}=c/2$ and $\beta_{d_c}=c/2$ in (\ref{eq09}), reduces to (\ref{eqsd}) obtained in \cite{dehghan2016new}.
%In bi-regular code ensemble,  the degrees of all variable nodes are the same  and the degrees of check nodes are the same, so our result in that case count the number of cycles of length $c$.
%Let $G=(U\cup W,E)$ be a random bi-regular graph in which all the nodes in $U$ have the same degree $d_u$ and all the nodes in $W$ have the same degree $d_w$.
%Consider the ensemble of such graphs as the number of nodes tends to infinity. In this case, for a fixed even value $c$,  the expected value of $N_c$ is given by
%\begin{equation}
%E(N_c)\sim \dfrac{\Big((d_u-1)(d_w-1)\Big)^{c/2}}{c}\:.
%\end{equation}
%This is the same result as the result of \cite{dehghan2016new} for the number of cycles of length $c$ in bi-regular graph.
\end{rem}

In the context of finding the asymptotic multiplicity of SS, ABS and LETS structures with only one simple cycle, we are particularly concerned about the degrees of variable nodes involved in the cycle, as these degrees determine the number of unsatisfied check nodes, $b$, in the trapping set. The degrees of the check nodes in the cycle however, are of no consequence to the structure of the trapping set. In the following, we address this new counting problem by generalizing the result of Theorem~\ref{Th45} to the cases where some of the variable or check nodes can have degrees that are freely chosen from the degree distributions. 

\begin{theo}\label{Th57}
Consider the random ensemble of irregular bipartite graphs $\mathcal{G}$, described in Theorem~\ref{Th45}, and let $\{d_i\}_{i=1}^n$ and $\{d_i'\}_{i=1}^{n'}$ be 
the set of variable and check node degrees satisfying $\sum_{i=1}^n d_i =\sum_{i=1}^{n'} d_i'= |E(G)|$, for $G \in \mathcal{G}$.
% Let $\Delta$ be a fixed natural number, such that $\Delta \geq d_1 \geq d_2  \geq \ldots \geq d_n$, and $\Delta \geq d_1' \geq d_2'  \geq \ldots \geq d_{m}'$, where $\sum_{i=1}^n d_i =\sum_{i=1}^{m} d_i'= |E(G)|$.
%Consider the probability space $\mathcal{G}$ of all Tanner graphs with node set $(U,W)$, where $U=\{u_1, u_2, \ldots , u_n\}, W=\{w_1,w_2, \ldots, w_{m}\}$, and in which the degree of node $u_i$ is $d_i$ and the degree of node $w_i$ is $d_i'$. Suppose that the degree sets $\{d_i\}$ and $\{d_i'\}$ follow the distributions $\lambda(x) =\sum_{i=2}^{z} \lambda_i x^{i-1}$ and $\rho (x) = \sum_{i=2}^{w} \rho_i x^{i-1}$, respectively, and that the graphs in $\mathcal{G}$ are selected uniformly at random. Also, assume that $z$ and $w$ are fixed natural numbers and  the number of variable nodes $n$ and the number of check nodes $m$ tends to infinity.
Let $\vec{\alpha'_c}=(\alpha_{z'},\ldots, \alpha_{z}, \alpha')$ and  $\vec{\beta'_c}=(\beta_{w'},\ldots, \beta_{w}, \beta' )$ be sequences of nonnegative integers such that $\alpha' + \sum_{i=z'}^z \alpha_i=c/2$ and $\beta' + \sum_{j=w'}^w \beta_j=c/2$. For vectors $\vec{\alpha_c'}$, $\vec{\beta_c'}$, and $G \in \mathcal{G}$, denote by $N_{\vec{\alpha_c'},\vec{\beta_c'}}(G)$ the number of cycles of length $c$ in $G$ such that each cycle has at least $\alpha_i$ variable nodes with degree $i$, for each $i=z',\ldots, z$, and at least $\beta_j$ check nodes with degree $j$, for each $j=w',\ldots, w$.\footnote{The cycle thus has $\alpha'$ variable nodes and $\beta'$ check nodes, whose degrees are freely selected from the degree distributions.} Then, as the size of the graph tends to infinity, for any fixed even value of $c \geq 4$, we have %the random variables $N_{\vec{\alpha_c'},\vec{\beta_c'}}$, has the expected value
\begin{align*}\label{eq217}
\mathbf{E}(N_{\vec{\alpha_c'},\vec{\beta_c'}})& \approx \displaystyle {{c/2}\choose{\alpha_{z'},\ldots, \alpha_{z}, \alpha'}}\prod_{i=z'}^{z}(\lambda_i(i-1))^{\alpha_i} {{c/2}\choose{\beta_{w'},\ldots, \beta_{w},\beta'}}\prod_{j=w'}^{w}(\rho_i(j-1))^{\beta_j} \\
& \times \Big(\frac{2}{|E(G)|}\displaystyle\sum_{i=1}^{n}{{d_i}\choose {2}}\Big)^{\alpha'}\Big(\frac{2}{|E(G)|}\displaystyle\sum_{i=1}^{n'}{{d_i'}\choose {2}}\Big)^{\beta'}\dfrac{1}{c}\:,\numberthis
\end{align*}
where the approximation is an asymptotic upper bound within the fixed multiplicative factor of $S(h_u)^{-\alpha'} \times S(h_w)^{-\beta'}$ from the exact value, with Specht's ratio $S(h)$ defined by $S(h) = \frac{(h-1)h^{\frac{1}{h-1}}}{e\log h}$,
for $h \neq 1$, and $S(1) = 1$ ($e$ is Euler's constant), and $h_u = \frac{z(z-1)}{z'(z'-1)},  h_w = \frac{w(w-1)}{w'(w'-1)}$.
\end{theo}

\begin{proof}
Similar to the proofs of Theorems~\ref{T1} and \ref{Th45}, here also, we work in the random ensemble of configurations.
%The proof of this theorem is a combination of the methods that were presented in the proof of Theorem \ref{Th1} and in the proof of Theorem 1 in \cite{dehghan2016new}.
%We partition $U$ (variable nodes) based on the degree of nodes. 
For each $i$, $z' \leq i \leq z$, let $U_i$ be the set of variable nodes with degree $i$. Also, for each $j$, $w' \leq j \leq w$, let $W_j$ be the set of check nodes with degree $j$.
To form an ($\vec{\alpha}_c',\vec{\beta}_c'$)-cycle, for each $i$, $z' \leq i \leq z$, one needs to choose $\alpha_i$ bins from $U_i$. 
Next, from the remaining unchosen nodes in $U$, one needs to choose $\alpha'$ bins. Similar selection process will need to be performed on the check node side of the graph.
%Similarly, for each $j$, $2\leq j \leq w$, we needs to choose $\beta_j$ bins from $W_j$. Next,  from the remaining unchosen nodes in $W$ one needs to choose $\beta'$ bins. 
Since the size of the graph tends to infinity and since $c$ and thus all $\alpha_i$ and $\beta_j$ values are constants, in the selection of $\alpha'$ variable nodes and $\beta'$ check nodes, we can 
assume that the number of choices is asymptotically equal to the case where these nodes are selected from the whole set of variable and check nodes, respectively.
%The number of variable nodes and check nodes tends to infinity, so the number of variable nodes of each degree in $G$ tends to infinity, also the number of check nodes of each degree in $G$ tends to infinity. On the other hand, the numbers $\alpha'$ and $\beta'$ are constant numbers. So, the number of ways that we can choose $\alpha'$ bins from  the remaining unchosen nodes in $U$ is approximately equal to the number of ways of choosing $\alpha'$ bins from the whole $U$. We have the same argument for choosing check nodes. 
Therefore, by taking similar steps as those taken to derive (\ref{E1}), we have the following equation for the number of ($\vec{\alpha}_c',\vec{\beta}_c'$)-cycles in a configuration:

\begin{align*}
 {\cal C}_{\vec{\alpha_c'},\vec{\beta_c'}} &\sim \Big( \displaystyle\sum_{\sigma_{z'} \subset U_{z'}, |\sigma_{z'}|=\alpha_{z'} \atop {\vdots \atop {\sigma_{z} \subset U_z, |\sigma_{z}|=\alpha_{z} \atop \sigma' \subset U, |\sigma'|=\alpha'} } }\Big[\prod_{i=z'}^{z} (i(i-1))^{\alpha_i} \prod_{u_i \in \sigma'} (d_i)(d_i-1)\Big] \Big)\\
 &\times \Big( \displaystyle\sum_{\sigma_{w'} \subset W_{w'}, |\sigma_{w'}|=\beta_{w'} \atop  {\vdots \atop {\sigma_{w} \subset W_w, |\sigma_{w}|=\alpha_{w} \atop \sigma' \subset W, |\sigma'|=\beta' } } }\Big[\prod_{j=w'}^{w} (j(j-1))^{\beta_j}\prod_{w_i\in \sigma}(d_i')(d_i'-1) \Big]\Big)\Big( \dfrac{(\frac{c}{2})!(\frac{c}{2})!}{c} \Big)\:, \numberthis \label{EN7}
\end{align*}
where $d_i = d(u_i)$ and $d_i' = d(w_i)$. By substituting (\ref{E4}) in (\ref{EN7}), we obtain

\begin{align*}
 {\cal C}_{\vec{\alpha_c'},\vec{\beta_c'}} & \sim \Big( (\dfrac{n}{\sum_{j}\lambda_j/j})^{{c/2-\alpha'}}\prod_{i=z'}^{z} \dfrac{(\lambda_i(i-1))^{\alpha_i}}{(\alpha_i)!}\displaystyle\sum_{\sigma'\subset U \atop |\sigma'|=\alpha'}\prod_{u_i\in \sigma'} (d_i)(d_i-1) \Big) \\
 & \times \Big((\dfrac{n}{\sum_{j}\rho_j/j})^{c/2-\beta'}\prod_{i=w'}^{w} \dfrac{(\rho_i(i-1))^{\beta_i}}{(\beta_i)!}
 \displaystyle\sum_{\sigma'\subset W \atop |\sigma'|=\beta'}\prod_{w_i\in \sigma'}(d_i')(d_i'-1) \Big)\Big(\dfrac{(\frac{c}{2})!(\frac{c}{2})!}{c} \Big)\:. \numberthis \label{EN8}
\end{align*}

In the following, we derive asymptotic upper and lower bounds on (\ref{EN8}) that differ only in a multiplicative factor. For this, we first use Maclaurin's inequality (see \cite{MR3015124}, pp 117-119), as described below.
Let $a_1, a_2, \ldots, a_n$ be positive real numbers, and for $k = 1, 2, \ldots, n$, define the averages $S_k$ as follows:

$$S_k = \frac{\displaystyle \sum_{ 1\leq i_1 < \cdots < i_k \leq n}a_{i_1} a_{i_2} \cdots a_{i_k}}{\displaystyle {n \choose k}},$$
where the summation is over all distinct sets of $k$ indices. Maclaurin's inequality then states:
$$S_1 \geq \sqrt{S_2} \geq \sqrt[3]{S_3} \geq \cdots \geq \sqrt[n]{S_n}\:.$$
Using Maclaurin's inequality, we thus have
\begin{equation}
(\prod_{j=1}^n a_j)^{\frac{1}{n}} \leq \sqrt[\alpha']{S_{\alpha'}} \leq \frac{\sum_{j=1}^n a_j}{n}\:.
\label{eq30}
\end{equation}
%$S_1 \geq \sqrt[\alpha']{S_{\alpha'}}$. By this inequality and (\ref{EN10}), we have:
On the other hand~\cite{specht1960theorie},
\begin{equation}
\frac{1}{S(h)} \times \frac{\sum_{j=1}^n a_j}{n} \leq (\prod_{j=1}^n a_j)^{\frac{1}{n}}\:,
\label{eq31}
\end{equation}
where $h = \frac{M}{m} (\geq 1)$, with $M$ and $m$ equal to the maximum and minimum values of numbers $a_1, a_2, \ldots, a_n$, and Specht's ratio $S(h)$ is defined by
\begin{equation}\label{FNE3}
S(h)=\dfrac{(h-1)h^{\frac{1}{h-1}}}{e \log h} \:\:\text{for}\:\: h \neq 1, \text{ and } S(1)=1\:,
\end{equation}
in which $e$ is Euler's number.

Combining (\ref{eq30}) and (\ref{eq31}), we have
\begin{equation}\label{ineq1}
{\displaystyle S(h)^{-\alpha'} {n \choose {\alpha'}}}  \Big(\dfrac{\sum_{j=1}^{n}a_j}{n}\Big)^{\alpha'} \leq \displaystyle \sum_{ 1\leq i_1 < \cdots < i_{\alpha'} \leq n}a_{i_1} a_{i_2} \cdots a_{i_{\alpha'}} \leq {\displaystyle {n \choose {\alpha'}}}  \Big(\dfrac{\sum_{j=1}^{n}a_j}{n}\Big)^{\alpha'}\:.
\end{equation}
Now, let $a_{i_k}=d_k (d_k-1)$. Focusing on the upper bound in (\ref{ineq1}), we then have

\begin{align*}
\displaystyle\sum_{\sigma'\subset U \atop |\sigma'|=\alpha'}\prod_{u_i\in \sigma'} (d_i)(d_i-1)
& \leq \displaystyle{{n}\choose {\alpha'}} \Big(\dfrac{\sum_{u_i\in U} (d_i)(d_i-1)}{n}\Big)^{\alpha'}  \\% \numberthis \label{E102}\\
&\sim \dfrac{n^{\alpha'}}{\alpha'!} \Big(\dfrac{\sum_{u_i\in U} (d_i)(d_i-1)}{n}\Big)^{\alpha'}\\
&= \dfrac{1}{\alpha'!} \Big(2\sum_{u_i\in U}  {{d_i}\choose {2}}\Big)^{\alpha'}. \numberthis \label{E92}
\end{align*}
Similarly, we can establish the following asymptotic lower bound: 
\begin{equation}\label{eq2w}
\displaystyle\dfrac{S(h_u)^{-\alpha'}}{\alpha'!} \Big(2\sum_{u_i\in U}  {{d_i}\choose {2}}\Big)^{\alpha'} \leq \displaystyle\sum_{\sigma'\subset U \atop |\sigma'|=\alpha'}\prod_{u_i\in \sigma'} (d_i)(d_i-1)\:.
\end{equation}
Now, by applying (\ref{E92}) and (\ref{eq2w}) to (\ref{EN8}) for both sides of the graph, we obtain the asymptotic upper and lower bounds on $\mathbf{E}(N_{\vec{\alpha_c'},\vec{\beta_c'}}) = {\cal C}_{\vec{\alpha_c'},\vec{\beta_c'}} (|E(G)|-c)! / |E(G)|!  \sim {\cal C}_{\vec{\alpha_c'},\vec{\beta_c'}} / |E(G)|^c$. In particular, the asymptotic value of the upper bound for (\ref{EN8}), being used to obtain the approximate value of $\mathbf{E}(N_{\vec{\alpha_c'},\vec{\beta_c'}})$, is calculated as follows:
\begin{align*}
{\cal C}_{\vec{\alpha_c'},\vec{\beta_c'}} & \approx \Big[ (\dfrac{n}{\sum_{j}\lambda_j/j})^{{c/2-\alpha'}}\prod_{i=z'}^{z} \dfrac{(\lambda_i(i-1))^{\alpha_i}}{(\alpha_i)!}
\times \frac{1}{\alpha'!} \Big(2\sum_{u_i\in U}  {{d_i}\choose {2}}\Big)^{\alpha'}
%\displaystyle{{n}\choose {\alpha'}}\Big(\dfrac{\sum_{u_i\in U} (d_i)(d_i-1)}{n}\Big)^{\alpha'}
\Big] \\
& \times \Big[(\dfrac{n}{\sum_{j}\rho_j/j})^{c/2-\beta'}\prod_{i=w'}^{w} \dfrac{(\rho_i(i-1))^{\beta_i}}{(\beta_i)!}
%\displaystyle{{n'}\choose {\beta'}}\Big(\dfrac{\sum_{w_i\in W} (d_i')(d_i'-1)}{m}\Big)^{\beta'}
\times \frac{1}{\beta'!} \Big(2\sum_{w_i\in W}  {{d_i'}\choose {2}}\Big)^{\beta'}
\Big]\Big(\dfrac{(\frac{c}{2})!(\frac{c}{2})!}{c} \Big)\\
& \sim \Big[ {{c/2}\choose{\alpha_{z'},\ldots, \alpha_{z}, \alpha'}}(\dfrac{n}{\sum_{j}\lambda_j/j})^{{c/2-\alpha'}}\prod_{i=z'}^{z} (\lambda_i(i-1))^{\alpha_i} \displaystyle
\Big( 2\sum_{u_i\in U}  {{d_i}\choose {2}}\Big)^{\alpha'}
\Big] \\
& \times \Big[ {{c/2}\choose{\beta_{w'},\ldots, \beta_{w},\beta'}}(\dfrac{n}{\sum_{j}\rho_j/j})^{c/2-\beta'}\prod_{i=w'}^{w} (\rho_i(i-1))^{\beta_i}
\displaystyle
\Big( 2\sum_{w_i\in W}  {{d_i'}\choose {2}}\Big)^{\beta'}
\Big]\Big(\dfrac{1}{c} \Big)\:. \numberthis \label{EN11}
\end{align*}

On the other hand, by (\ref{EN4}), we have
\begin{align*}
|E(G)|^{c/2 - \alpha'} = \Big( \dfrac{n}{\sum_{j}\lambda_j/j} \Big) ^{c/2-\alpha'}\:.  \numberthis \label{EN12}
\end{align*}

This together with (\ref{EN11}) and $\mathbf{E}(N_{\vec{\alpha_c'},\vec{\beta_c'}}) \sim {\cal C}_{\vec{\alpha_c'},\vec{\beta_c'}} / |E(G)|^c$ result in (\ref{eq217}).

%\begin{align*}
%\mathbf{E}(N_{\vec{\alpha_c'},\vec{\beta_c'}})&\approx \displaystyle {{c/2}\choose{\alpha_{z'},\ldots, \alpha_{z}, \alpha'}}\prod_{i=2}^{z}(\lambda_i(i-1))^{\alpha_i} {{c/2}\choose{\beta_{w'},\ldots, \beta_{w},\beta'}}\prod_{j=2}^{w}(\rho_i(j-1))^{\beta_j} \\
%& \times \Big(\frac{2}{|E(G)|}\displaystyle\sum_{i=1}^{n}{{d_i}\choose {2}}\Big)^{\alpha'}\Big(\frac{2}{|E(G)|}\displaystyle\sum_{i=1}^{n'}{{d_i'}\choose {2}}\Big)^{\beta'}\dfrac{1}{c} \numberthis \label{EN16}
%\end{align*}

The asymptotic lower bound on $\mathbf{E}(N_{\vec{\alpha_c'},\vec{\beta_c'}})$ is equal to the upper bound of (\ref{eq217}) multiplied by $S(h_u)^{\alpha'} \times S(h_w)^{\beta'}$, proving the claim about the accuracy of the approximation.
\end{proof}

\begin{rem}
In Theorem \ref{Th57}, consider the case where $\alpha'=\beta'=c/2$. In this case, we are interested in cycles of length $c$ without any constraint on the node degrees. One can see that in this case, Equation (\ref{eq217}) reduces to (\ref{eq12}), obtained in \cite{dehghan2016new}.
\end{rem}

\begin{cor}
For the ensemble of irregular bipartite graphs discussed in Theorem~\ref{Th57}, let $d_{v_{\min}} =  2$. Then, the expected multiplicity of local $(a,0)$ stopping sets
is given by 
$$
\mathbf{E}(N_{(a,0)}^{SS}) \approx \Big(\frac{2 \lambda_2}{|E(G)|}\displaystyle\sum_{i=1}^{n'}{{d_i'}\choose {2}}\Big)^{a} \times \dfrac{1}{2a}\:.
$$
\label{cor543}
\end{cor}
\begin{proof}
Based on Theorem~\ref{thm54}, we need to find the asymptotic expected number of simple cycles of length $2a$ that consist of $a$ variable nodes of degree $2$. This can be estimated from Theorem~\ref{Th57} by counting the number of cycles of length $c=2a$ that consist of $a$ degree-$2$ variable nodes and $a$ check nodes with free degrees, i.e., $\alpha_2 = a$ and $\beta'=a$, and noting that in the asymptotic regime, all such cycles are simple.
\end{proof}

\begin{cor}\label{cor329}
For an ensemble of irregular graphs with $d_{v_{\min}}=3$, the asymptotic average multiplicity of $(a,a)$ ABSs is given by
$$
\mathbf{E}(N_{(a,a)}^{ABS}) \approx \Big(\frac{4 \lambda_3}{|E(G)|}\displaystyle\sum_{i=1}^{n'}{{d_i'}\choose {2}}\Big)^{a} \times \dfrac{1}{2a}\:.
$$
\end{cor}
\begin{proof}
Based on Theorem~\ref{proABS}, the only ABS structures in the $(a,a)$ class whose asymptotic expected multiplicity is non-zero, are simple cycles consisting only of degree-$3$ variable nodes. We can thus estimate the number of such cycles using Theorem~\ref{Th57} by counting cycles of length $c=2a$ that consist of $a$ degree-$3$ variable nodes, and $a$ check nodes of free degree, i.e., $\alpha_3=a$ and $\beta'=a$, and the fact that in the asymptotic regime, all such cycles are simple.
\end{proof}

\begin{cor}\label{cor330}
For an ensemble of irregular graphs with $d_{v_{\min}}=2$, the asymptotic average multiplicity of local $(a,b)$ ABSs for an arbitrary value of $a$ and different values of $b \leq a$, is given by
$$
\mathbf{E}(N_{(a,b)}^{ABS}) \approx \frac{(a-1)!}{2 \times b! \times (a-b)!} \Big(\frac{2\lambda_3}{\lambda_2}\Big)^b \Big(\frac{2 \lambda_2}{|E(G)|}\displaystyle\sum_{i=1}^{n'}{{d_i'}\choose {2}}\Big)^{a}\:. 
$$
\end{cor}
\begin{proof}
Based on Theorem~\ref{proABS}, the only ABS structures in the $(a,b)$ class whose asymptotic expected multiplicity is non-zero, are simple cycles consisting only of degree-$2$ and degree-$3$ variable nodes. We can thus estimate the number of such cycles using Theorem~\ref{Th57} by counting cycles of length $c=2a$ that consist of $\alpha_2=a-b$ degree-$2$ and $\alpha_3 =b$ degree-$3$ variable nodes, and $a$ check nodes of free degree, i.e., $\beta'=a$, and the fact that in the asymptotic regime, all such cycles are simple. The result of the corollary then follows from Equation~(\ref{eq217}) after some simplifications. 
\end{proof}

Based on the discussions in Subsection~\ref{SS4-2}, we know that the only LETS structures that have finite non-zero multiplicity asymptotically are those with only a simple cycle. In Theorem~\ref{Th57}, we derived an asymptotic approximation for the expected number of such cycles with different lengths and with different combination of variable degrees. The following proposition, whose proof is straight forward,  along with the results of Theorem~\ref{Th57} can be used to estimate the asymptotic multiplicity of LETS structures in different classes for irregular codes.
%It was shown in \cite{dehghan2017} that for an LDPC code, if the  code ensemble is irregular, then for every structure  $G(S)$ in LETSs, we have $|V(G(S))|\leq|E(G(S))|$. In other words, every LETS has at least one cycle. On the other hand, it was shown in \cite{dehghan2017} that if an structure has at least two cycles, then the average number of that structure tends to zero as the length of the code tends to infinity. So, for a give $a$ and $b$, we are interested in the number of $(a,b)$ LETSs such that each of them has exactly one cycle. In the following theorem, for the first time, we present the connection between dominant LETSs and the short cycles with low ACE number in the code’s Tanner graph.

\begin{pro}\label{Pro37}
Consider an ensemble ${\cal G}$ of irregular LDPC codes, and suppose that constants $a$ and $b$ are such that the asymptotic expected multiplicity of $(a,b)$ LETSs in ${\cal G}$ is a finite non-zero value.
Then, the asymptotic average multiplicity of $(a,b)$ LETSs in ${\cal G}$ is equal to the asymptotic expected number of simple cycles of length $2a$ in ${\cal G}$ whose variable node degrees add up to $2a+b$. 
%Let $C$ be a random irregular LDPC code such that for each variable node $v$ its degree is at least two and  $a_{\max}$ be a constant number. Then, for each $a$, $2 \leq a\leq a_{\max} $, the average number of  $(a,b)$ LETSs in $C$ is equal  to the  number of chordless cycles of length $2a$ in the code ensemble such that in each cycle the sum of degrees of variable nodes is equal to $ 2a+b$.
\end{pro}

%\begin{proof}
%For each structure in the class of $(a,b)$ LETSs, if that structure has at least two cycles, then its average number tends to zero. Let, $G(S)$ be an structure in the class of $(a,b)$ LETSs such that it has exactly one cycle. The structure is leafless, so the length of the cycle is $a$ (i.e. every variable node is in the cycle). Thus, the number of odd degree check nodes is equal to the sum of the degrees of variable nodes minus $2a$ (i.e. the sum of the degrees of variable nodes is equal to $2a+b$).
%\end{proof}

%\begin{rem}
%For a given $a$ and $b$, we can use Theorem \ref{Th2} and Theorem \ref{Th3}, to obtain the average number of $(a,b)$ LETSs.
%\end{rem}
The following example explains how the combination of Theorem~\ref{Th57} and Proposition~\ref{Pro37} can be used to estimate the multiplicity of LETS structures in different classes for irregular codes.
%\begin{rem}
%Theorem \ref{Th3} shows that there is a close relationship between the short cycles with low ACE number in the code’s Tanner graph and dominant LETSs. Also, Theorem \ref{Th2} shows that in some cases (by attentions to degree distributions) the ration of dominant LETSs of size $a$ to all LETSs of size $a$ is a small number. Thus, in irregular code ensemble, removing short cycles with low ACE number in the code’s Tanner graph instead of increasing the girth is a good method to remove dominant LETSs.
%\end{rem}
\begin{ex}
Consider the ensemble of irregular LDPC codes with degree distributions $\lambda(x)=\lambda_i x^{i-1} + \lambda_j x^{j-1}$ and $\rho(x)=x^{k-1}$, where $i \geq 2$, $j > i$ and $k > 2$ are constant integers.
For any fixed value of $a \geq 2$, any class of $(a,b)$ LETSs with the value of $b$ in the range $a \times (i-2) \leq b \leq a \times (j-2)$ has a non-zero finite expected value asymptotically. For each such class, with given 
$a$ and $b$ values, we can uniquely determine the number of variable nodes of each degree, $i$ and $j$, that exist in simple cycles belonging to the $(a,b)$ LETS class. For this, one needs to solve the following system of linear equations:
$i \alpha_i + j \alpha_j = 2a+b$ and $\alpha_i + \alpha_j = a$, to find the number of variable nodes of degrees $i$ and $j$ in the cycle. Solving these equations results in $\alpha_j = (2a+b-ia)/(j-i)$ and $\alpha_i = a - \alpha_j$. Using Theorem~\ref{Th57}, we can then estimate the asymptotic expected multiplicity
of $(a,b)$ LETSs by
\begin{equation}
\mathbf{E}(N_{(a,b)}^{LETS}) \approx \displaystyle {{a}\choose{\alpha_i,\alpha_j}}((i-1) \lambda_i)^{\alpha_i} ((j-1) \lambda_j)^{\alpha_j} (k-1)^{a} \times \dfrac{1}{2a}\:.
\label{eqex}
\end{equation}
\label{Ex29}
\end{ex}

\section{Asymptotic Expected Multiplicity of ETS Structures in Regular and Irregular Bipartite Graphs}
\label{sec-ets}

The only trapping sets with finite non-zero asymptotic expected multiplicity are those which contain a single (simple) cycle. One can thus see that in the asymptotic regime, ETS structures with finite non-zero expected multiplicity consist of a simple cycle and a collection of trees stemming from the variable nodes of that cycle. Each such tree is rooted at the corresponding variable node and contains check nodes of degrees two and one, where degree-$1$ check nodes form all the leaves of the tree. The problem of computing the multiplicity of ETSs with $a$ variable nodes can thus be broken into two subproblems: (i) calculating the multiplicity of simple cycles of different length $\ell$ ($\leq 2a$) with different variable degrees (if the graph is irregular), and (ii) calculation of the number of possible ways that a collection of trees with $a - \ell/2$ variable nodes (in addition to the roots) can be attached to the variable nodes of the cycle. To solve the first subproblem, we use the results of~\cite{dehghan2016new} and those obtained in Section~\ref{LETS-IRREG}. To tackle the second subproblem, we formulate the problem as a recursive counting problem whose solution is a generalization of Catalan numbers~\cite{MR1098222}. 

To formulate Subproblem (ii), we consider a simple cycle $C$ of length $\ell = 2(a-i)$, where $0 \leq i \leq a-g/2$. In this case, the number of variable nodes within the ETS but outside $C$ is $i$. For a given value of $i$, we partition $i$
into $a-i$ non-negative integers $k_1, \ldots, k_{a-i}$, where each $k_j, 1 \leq j \leq a-i$, satisfies $k_j \leq i$, and corresponds to one of the variable nodes of $C$.  The integer $k_j$, in fact, is used to denote the number of nodes outside $C$ that are in the tree rooted at the $j$-th variable node of $C$. We thus have the constraint $\sum_{j=1}^{a-i} k_j =i$. We further use the notation $t_p, 0 \leq p \leq i$, to denote the multiplicity of $k_j$'s with $|k_j|=p$. Thus, we have 
$\sum_{p=1}^i p \times t_p = i$, and $t_0 = a - i - \sum_{p=1}^i t_p$. We use the notation  $(t_0, t_1,  \ldots, t_{i} \bigtriangledown a)$ for any sequence on non-negative integers $t_0, \ldots, t_i$ that satisfy these two equations. For example, 
$(8,1,0,0,1,0 \bigtriangledown 15)$.  

%Without loss of generality, assume that $G(S)$ has a cycle $\mathcal{C}$ of length $\ell$, where $g\leq \ell \leq 2a$. When we remove the edges of $\mathcal{C}$ from the graph $G(S)$, the remaining graph is a disjoint union of trees (i.e. forest) such that the sum of nodes of these trees is equal to the number of nodes of $G(S)$. We call this forest $\mathcal{F}(S)$. Define the {\it v--size} of a tree, as the number variable nodes in that tree minus one if it has at least one variable node and zero otherwise (in other words, ``v-size"   stands for the ``variable-size" of the tree if we exclude the root). So, in the forest $\mathcal{F}(S)$, the sum of the v--sizes of its trees is equal to $a-\ell/2$. Let $t_0, t_1,  \ldots, t_{i}$ be $i+1$ nonnegative integers, we say that the sequence of  $i+1$ numbers $t_0, t_1,  \ldots, t_{i}$ is a {\em valid} sequence of v--sizes for $(a,(d_v-2)a)$ ETSs with a cycle of length $2(a-i)$, and denote it by $ (t_0, t_1,  \ldots, t_{i} \bigtriangledown a)$ if $\sum_{j=1}^{i} j \times t_j = i$  and $a- i=\sum_{j=0}^{a-i}  t_j$. For instance, the sequence  $ (8,1,0,0,1,0 \bigtriangledown 15)$ is a   valid  sequence of v--sizes (for $i=5$ and $a=15$). 
%To find all  valid  sequences of v--sizes for a given $i$ and $a$, we should find all possible partitions of number $i$ into at most $a-i$ positive integers. Let $h_1,h_2, \ldots, h_{a-i}$ be a  partition of number $i$ into $a-i$ positive integers, for each $i$, $i\geq 1$, let $t_i$ be  the number of $i$ in the set $\{h_1,h_2, \ldots, h_{a-i}\}$). Also, let  $t_0= a-i-\sum_{j=1}^{a-i}  t_j$.
 
In the following, we first consider the case of biregular graphs and then generalize the results to variable-regular and irregular cases.

\subsection{Biregular Bipartite Graphs}

\begin{theo} \label{Th95}
Consider the ensemble of $(d_v,d_c)$-regular graphs in which the number of nodes tends to infinity. Given fixed values of $a$, $d_c$ and $g$, for $d_v=3$, we have
%Let $G=(U\cup W,E)$ be a random $(d_v,d_c)$-regular graph. Consider the ensemble of such graphs as the number of nodes tends to infinity and denote the number of $(a,(d_v-2)a)$ ETSs  in that ensemble by $D_{a,d_v,d_c}$. For a fixed values $a$ and $d_v$, the expected value of $D_{a,d_v,d_c}$, for the case $d_v=3$ is given by the following asymptotic equality
\begin{align*}\label{eq47}
\mathbf{E}(N_{(a,a)}^{ETS})&\sim\sum_{i=0}^{a-g/2} \sum_{{t_0, t_1,\ldots, t_{i} }\atop {(t_0, t_1,  \ldots, t_{i} \bigtriangledown a)}}\Big[ \dfrac{ 2^{a-i}(d_c-1)^{a } }{2(a-i) }{{a-i}\choose{t_0, t_1,t_2, \ldots, t_{i} }}   \prod_{j=1}^{i} (\frac{1}{j}{{2j }\choose{j-1}})^{t_j}\Big],\numberthis
\end{align*}
and for a constant $d_v \geq 4$, we have %the following   asymptotic approximation
\begin{equation}
\label{eq48}
\mathbf{E}(N_{(a,(d_v-2)a)}^{ETS})  \approx 
\sum_{i=0}^{a-g/2} \sum_{{t_0, t_1,\ldots, t_{i} }\atop {(t_0, t_1,  \ldots, t_{i} \bigtriangledown a)}}\Big[ \dfrac{ (d_v-1)^{a-i}(d_c-1)^{a }}{2(a-i)}  {{a-i}\choose{t_0, t_1,t_2, \ldots, t_{i} }}    \prod_{j=1}^{i} (\frac{ 1}{j+1}{{(d_v-1)(j+1) }\choose{j}} )^{t_j}   \Big]\:,
\end{equation}
where the asymptotic approximation is an upper bound within a multiplicative factor of $(a-g/2+1)^{-a}$ from the exact value.
\end{theo}

\begin{proof}{
We first note that based on Theorem~\ref{C2}, the only classes of ETSs with finite non-zero asymptotic expected multiplicity are those with $b=a(d_v-2)$. These classes, which consist of all the structures that contain a single (simple) cycle, are the focus of this theorem. 
In the asymptotic regime, any $(a, a(d_v-2))$ ETS ${\cal S}$ contains a simple cycle $C$ whose neighborhood within any constant radius is a forest with probability one. 
The structure of this forest follows the degree distribution of the ensemble, i.e., any tree rooted at one of the variable nodes of $C$ is connected to $d_v-2$ check nodes and each such check node is connected to $d_c-1$ other variable nodes, which in turn are each connected to $d_v-1$ other check nodes and so on. The ETS ${\cal S}$, in addition to $C$, consists of some trees, each  rooted at one of the variable nodes of $C$. Consider a variable node $v$ of $C$, which is the root of a tree $T_v$ within ${\cal S}$. The tree $T_v$ is a subgraph of the neighborhood tree of $C$ rooted at $v$. We refer to the collection of trees such as $T_v$, which are part of the ETS ${\cal S}$, and each are rooted at one of the variable nodes of $C$, as the {\em forest of} ${\cal S}$. Each tree $T_v$ in this forest, is called a {\em basic tree}. The number of variable nodes other than the root in a basic tree is called the {\em variable-size} of the tree or ``v-size,'' in brief. Variable-size zero for a basic tree $T_v$ implies that variable node $v$ is connected to $d_v-2$ check nodes of degree-$1$ within the ETS. Basic trees have the following properties: (i) the root has degree $d_v-2$, (ii) the degree of check nodes is either two or one, (iii) the degree of all the variable nodes in the tree, except the root, is $d_v$, and (iv) the leaves of the tree are all check nodes (with degree one).

We use the notation $R(j,d_v,d_c)$ to denote the number of possible basic trees with constant v-size $j$ initiated from an arbitrary variable node of an arbitrary simple cycle within a Tanner graph from the ensemble of $(d_v,d_c)$-regular graphs in the asymptotic regime. One can then see that
\begin{align*}\label{F3}
N_{(a,(d_v-2)a)}^{ETS} &\sim\sum_{i=0}^{a-g/2} \sum_{{t_0,t_1,\ldots, t_{i} }\atop {(t_0,t_1,  \ldots, t_{i} \bigtriangledown a)}}\Big[ N_{2(a-i)} {{a-i}\choose{t_0,t_1,t_2, \ldots, t_{i}}}  \times \prod_{j=1}^{i}   R(j,d_v,d_c)^{t_j} \Big]\:.\numberthis
\end{align*}
In (\ref{F3}), the outer summation is to count simple cycles of different length $2(a-i)$, and the inner summation is responsible for counting all the possible ways that the $i$ variable nodes that are not part of the cycle can be partitioned into subsets of different size $j$. The multinomial coefficient describes the number of ways that the $a-i$ variable nodes of the simple cycle can be connected to basic trees of different v-sizes.

To calculate $R(j,d_v,d_c)$, we note that any of the $j$ variable nodes in a basic tree has a check node of degree $2$ as its parent. We thus have
\begin{equation}\label{F01}
R(j,d_v,d_c)=(d_c-1)^{j} \times B_{j+1}^{d_{v}}\:,
\end{equation}
where $B_{j+1}^{d_{v}}$ denotes the number of possible rooted trees with $j+1$ nodes such that in each tree, the root has degree at most $d_v - 2$, and all the other nodes have degree at most $d_v$.
The term  $(d_c-1)^{j}$ in (\ref{F01}) accounts for the fact that at each non-leaf check node of a basic tree, one has $d_c-1$ choices to select the variable node child. 

By replacing (\ref{F01}) in (\ref{F3}), we obtain
\begin{align*}\label{F60}
N_{(a,(d_v-2)a)}^{ETS} &\sim\sum_{i=0}^{a-g/2} \sum_{{t_0,t_1,\ldots, t_{i} }\atop {(t_0,t_1,  \ldots, t_{i} \bigtriangledown a)}}\Big[ N_{2(a-i)} {{a-i}\choose{t_0,t_1,t_2, \ldots, t_{i}}}  \times (d_c-1)^{i} \prod_{j=1}^{i}\Big(   B_{j+1}^{d_{v}}\Big)^{t_j} \Big]\:.\numberthis
\end{align*}
Taking the expected value of the above expression and replacing $\mathbf{E}(N_{2(a-i)})$ in the expected value with (\ref{eqsd}), we have
\begin{align*}\label{F8}
\mathbf{E} (N_{(a,(d_v-2)a)}^{ETS})&\sim\sum_{i=0}^{a-g/2} \sum_{{t_0, t_1,\ldots, t_{i} }\atop {(t_0, t_1,  \ldots, t_{i}\bigtriangledown a)}}\Big[ \dfrac{ (d_v-1)^{a-i}(d_c-1)^{a }}{2(a-i)}  {{a-i}\choose{t_0, t_1,t_2, \ldots, t_{i} }}  \prod_{j=1}^{i} \Big( B_{j+1}^{d_{v}}\Big)^{t_j }  \Big]\:.\numberthis
\end{align*}

By the definition of $B_{j+1}^{d_{v}}$, we have
\begin{equation}\label{F02}
B_{j+1}^{d_{v}}=\sum_{ {k_1,\ldots,k_{d_{v}-2}\geq 0}   \atop {k_1+\cdots+k_{d_{v}-2}=j } }
C_{k_1 }^{d_v} \times C_{k_2 }^{d_v} \times \cdots \times C_{k_{d_{v}-2}}^{d_v},
\end{equation}
where $C_{j}^{d_v}$ is the number of rooted trees with $j$ nodes such that the root has degree at most $d_v-1$ and each other node has degree at most $d_v$. To calculate  $C_{j}^{d_v}$,
we can then use the following recursion
\begin{equation}
C_{j}^{d_{v}}=\sum_{ {k_1,\ldots,k_{d_{v}-1}\geq 0}   \atop {k_1+\cdots+k_{d_{v}-1}=j-1 } }
C_{k_1 }^{d_v} \times C_{k_2 }^{d_v} \times \cdots \times C_{k_{d_{v}-1}}^{d_v}\:, 
\label{eq27}
\end{equation}
with base cases $C_{0}^{d_{v}}=C_{1}^{d_{v}}=1$.
%
%When we start from a cycle of length $2(a-i)$, then the sum of the v--sizes of trees is $i$. Thus, by (\ref{EZ1}), we have:
%\begin{equation}\label{F7}
%\prod_{j=1}^{i}  \Big( (d_c-1)^{j}   B_{j+1}^{d_{v}}\Big)^{t_j}=(d_c-1)^{i} \prod_{j=1}^{i}\Big(   B_{j+1}^{d_{v}}\Big)^{t_j}.
%\end{equation}
%
%\begin{figure}[ht]
%	\begin{center}
%		\includegraphics[scale=.2]{Fig002}
%		\caption{An example of tree $T'$ ($T'$ is a subgraph of $G'$) and all of its corresponding trees in $G$.
%		} \label{Fig002}
%	\end{center}
%\end{figure}
%
%\underline{Calculation of $B_{j+1}^{d_{v}}$:}
%Before, calculating the exact value of $B_{j+1}^{d_{v}}$, let us to introduce another parameter.
%Let $T$ be a  tree with the root $v$, such that:\\
%$(i)$ the height of the tree is $j$,\\
%$(ii)$ the degree of the node $v$ is $d_{v}-1$, and\\
%$(iii)$ the degree of each node, except the leaves and the node $v$ is $d_v$.\\
%Let $C_{j}^{d_v}$ be the the number  of subtrees of $T$ such that: $(i)$ each subtree has exactly $j$ nodes, $(ii)$ and each subtree  has the node $v$.
%
%Now, we calculate the exact value of $C_{j}^{d_v}$. We have the following recursive relation:
%$$C_{j}^{d_{v}}=\sum_{ {k_1,\ldots,k_{d_{v}-1}\geq 0}   \atop {k_1+\cdots+k_{d_{v}-1}=j-1 } }
%C_{k_1 }^{d_v} \times C_{k_2 }^{d_v} \times \cdots \times C_{k_{d_{v}-1}}^{d_v}, $$
%where $C_{0}^{d_{v}}=C_{1}^{d_{v}}=1$. 
The solution to the recursion (\ref{eq27}) is a generalization of Catalan numbers \cite{MR1098222}, and is given by
\begin{equation}\label{F4}
C_{j}^{d_{v}}=\frac{1}{j}{{j(d_v-1)}\choose{j-1}}.
\end{equation}

%Now, we can calculate $B_{j+1}^{d_{v}}$. Note that since the distance between every two cycle is a large number, in our calculation, we can consider the graph $G'\setminus E(\mathcal{C}')$ as a tree. Thus,
%\begin{equation}\label{F02}
%B_{j+1}^{d_{v}}=\sum_{ {k_1,\ldots,k_{d_{v}-2}\geq 0}   \atop {k_1+\cdots+k_{d_{v}-2}=j } }
%C_{k_1 }^{d_v} \times C_{k_2 }^{d_v} \times \cdots \times C_{k_{d_{v}-2}}^{d_v},
%\end{equation}
%where $C_{j}^{d_v}$ is the $j$th number in (\ref{F4}).
For the case $d_v=3$, by (\ref{F02}), we  have $B_{j +1}^{3}=C_{j }^{3}$, and thus,
\begin{equation}\label{F5}
B_{j+1}^{3}=\frac{1}{j}{{2j }\choose{j-1}}.
\end{equation}
%Taking the expected value of (\ref{F60}), replacing $\mathbf{E}(N_{2(a-i)})$ in the expected value with (\ref{eqsd}), and 
By replacing $B_{j +1}^{3}$ with (\ref{F5}) in (\ref{F8}), we obtain (\ref{eq47}).

To derive (\ref{eq48}), we first establish an upper and a lower bound on  $B_{j+1}^{d_{v}}$ in terms of $C_{j+1}^{d_{v}}$. For the upper bound, it is clear that
\begin{equation}\label{eq01}
B_{j+1}^{d_{v}} \leq C_{j+1}^{d_{v}}
\end{equation}
For the lower bound, by the definition of $C_{j+1}^{d_{v}}$, we have
\begin{equation}\label{eq05}
C_{j+1}^{d_{v}}= \sum_{i=0}^{j} B_{j+1-i}^{d_{v}} \times C_{i}^{d_{v}}\:.
\end{equation}
On the other hand, for each $i$ in the above summation, the following inequality holds:
\begin{equation}\label{eq06}
 B_{j+1-i}^{d_{v}} \times C_{i}^{d_{v}} \leq  B_{j}^{d_{v}}\:.
\end{equation}  
Combining (\ref{eq06}) with (\ref{eq05}), we obtain
\begin{equation}\label{eq72}
C_{j+1}^{d_{v}} \leq (j+1) \times B_{j}^{d_{v}}\:.
\end{equation}
This together with $B_{j}^{d_{v}} \leq B_{j+1}^{d_{v}}$ results in
\begin{equation} \label{eq009}
\dfrac{C_{j+1}^{d_{v}}}{j+1}\leq B_{j+1}^{d_{v}}\:.
\end{equation}
%Taking the expected value of (\ref{F60}), replacing $\mathbf{E}(N_{2(a-i)})$ in the expected value with (\ref{eqsd}), and 
By replacing $B_{j +1}^{d_v}$ with the upper bound of (\ref{eq01}) in (\ref{F8}), and using (\ref{F4}) for $C_{j+1}^{d_{v}}$,
we obtain the asymptotic upper bound of (\ref{eq48}).

To obtain the asymptotic lower bound that determines the accuracy of Estimate (\ref{eq48}), we apply the lower bound of (\ref{eq009}) on $B_{j+1}^{d_{v}}$ to (\ref{F8}). This results in the following lower bound
\begin{align*}
\sum_{i=0}^{a-g/2} \sum_{{t_0, t_1,\ldots, t_{i} }\atop {(t_0, t_1,  \ldots, t_{i} \bigtriangledown a)}}\Big[ \dfrac{ (d_v-1)^{a-i}(d_c-1)^{a }(a-i)! }{2(a-i)(t_0)!(t_1)!\cdots(t_{i})! }    \prod_{j=1}^{i} (\frac{ 1}{(j+1)^2}{{(d_v-1)(j+1) }\choose{j}} )^{t_j}   \Big]\:.
\end{align*}
This lower bound can be further bounded from below by the following equation:
\begin{align*}\label{eq260}
\sum_{i=0}^{a-g/2} \sum_{{t_0, t_1,\ldots, t_{i} }\atop {(t_0, t_1,  \ldots, t_{i} \bigtriangledown a)}}\Big[ \dfrac{ (d_v-1)^{a-i}(d_c-1)^{a }(a-i)! }{2(a-i)(t_0)!(t_1)!\cdots(t_{i})! }    \prod_{j=1}^{i} (\frac{ 1}{j+1}{{(d_v-1)(j+1) }\choose{j}} )^{t_j}   \Big]  \times \frac{1}{(a-g/2+1)^a}\:.\numberthis
\end{align*}
To obtain (\ref{eq260}), we have used the following sequence of inequalities:
$$
\prod_{j=1}^i \frac{1}{(j+1)^{t_j}} \geq \prod_{j=1}^i \frac{1}{(i+1)^{t_j}} = \frac{1}{(i+1)^{\sum_{j=1}^i t_j}} \geq \frac{1}{(i+1)^{a-i}} \geq \frac{1}{(a-g/2+1)^a}\:.
$$
}\end{proof}

\begin{rem}
We note that (\ref{eq48}) is, in general, less accurate than (\ref{eq47}), in the sense that while (\ref{eq47}) is asymptotically equal to the true expected value $\mathbf{E} (N_{(a,(d_v-2)a)}^{ETS})$, Equation (\ref{eq48}) provides only  an asymptotic upper bound for $\mathbf{E} (N_{(a,(d_v-2)a)}^{ETS})$. It should, however, be noted that in (\ref{F8}), we present an asymptotic value for $\mathbf{E}(N_{(a,(d_v-2)a)}^{ETS})$ even when $d_v\geq 4$. 
To use (\ref{F8}), one needs to calculate the exact value of $B_{j+1}^{d_{v}}$ using (\ref{F02}) and (\ref{F4}). This can be easily performed for the cases where $j$ and $d_v$ are relatively small. In Table \ref{TAB1}, we have listed the values of $B_{j+1}^{d_{v}}$ for some small values of $d_v$ and $j$.  
\end{rem}

\begin{table}[h]
\centering
\caption{Values of $B_{j+1}^{d_{v}}$ for $d_v = 4, 5$, and $j=0, 1, 2, 3$.}
\label{TAB1}
	\begin{tabular}{ |c|c|c|c|c| } 
		\hline
		 &  $B_{1}^{d_{v}}$ & $B_{2}^{d_{v}}$ & $B_{3}^{d_{v}}$ & $B_{4}^{d_{v}}$ \\
		\hline
$d_v=4$	 & 1               &      2          &  7              &    30            \\ 
\hline 
$d_v=5$	& 1               &      3          &  15             &    91            \\
		\hline
		\end{tabular}
\end{table}

%\begin{center}
%	\begin{tabular}{ |c|c|c|c|c| } 
%		\hline
%		 &  $B_{1}^{d_{v}}$ & $B_{2}^{d_{v}}$ & $B_{3}^{d_{v}}$ & $B_{4}^{d_{v}}$ \\
%		\hline
%$d_v=4$	 & 1               &      2          &  7              &    30            \\ 
%\hline 
%$d_v=5$	& 1               &      3          &  15             &    91            \\
%		\hline
%		\end{tabular}\label{TAB1}
%\end{center}

\subsection{Variable-regular bipartite graphs}

\begin{theo} \label{NT3}
Consider the random ensemble $\mathcal{G}$ of variable-regular bipartite graphs with variable degree $d_v$ and check node degree distribution $\rho(x)$, corresponding to the set of check node degrees $\{d'_i\}_{i=1}^{n'}$, and let the number of nodes in the graphs tend to infinity.
%Let $G=(U\cup W,E)$ be a random variable-regular graph with variable degree $d_v$ and check degree distribution $\rho(x)$.
%Consider the ensemble of such graphs as the number of nodes tends to infinity and denote the number of $(a,(d_v-2)a)$ ETSs  in that ensemble by $D_{a,d_v,\rho(x)}$. 
In $\mathcal{G}$, for fixed values $a$, $d_v$, $d_{c_{\max}}$ and $g$, we have
\begin{align*}\label{eq200}
\mathbf{E}( N_{(a,a(d_v-2))}^{ETS} ) \approx \sum_{i=0}^{a-g/2} \sum_{{t_0, t_1,\ldots, t_{i} }\atop {(t_0, t_1,  \ldots, t_{i} \bigtriangledown a)}}
\Big[&
\dfrac{\displaystyle\Big( \frac{2(d_v-1)}{nd_v}\displaystyle\sum_{i=1}^{n'}{{d_i'}\choose {2}}\Big)^{a-i}}{2(a-i)} {{a-i}\choose{t_0,t_1,t_2, \ldots, t_{i} }} \times\\
&(\overline{d_c}-1)^{i}  \prod_{j=1}^{i} (\frac{ 1}{j+1}{{(d_v-1)(j+1) }\choose{j}} )^{t_j}   \Big]. \numberthis
\end{align*}
where $\overline{d_c}$ is the average check node degree, and the approximation is an asymptotic upper bound within the fixed multiplicative factor of 
$[(a-g/2+1) \times S(h_w)]^{-a} $ from the exact value, with $S(h)$ being Specht's ratio, defined in Theorem~\ref{Th57},  
%defined by $S(h)=\dfrac{(h-1)h^{\frac{1}{h-1}}}{e \log h}$ for $h \neq 1$, and $S(1)=1$ ($e$ is Euler's constant), 
and $h_w= \frac{d_{c_{\max}} (d_{c_{\max}}-1)}{d_{c_{\min}}(d_{c_{\min}}-1)}$.
\end{theo}

\begin{proof}
Using similar discussions as those in the proof of Theorem~\ref{Th95}, we obtain the following equation: 
\begin{align*}\label{F3000}
N_{(a,a(d_v-2))}^{ETS} &\sim\sum_{i=0}^{a-g/2} \sum_{{t_0, t_1,\ldots, t_{i} }\atop {(t_0, t_1,  \ldots, t_{i} \bigtriangledown a)}}\Big[ N_{2(a-i)} {{a-i}\choose{t_0, t_1,t_2, \ldots, t_{i} }}  \times \prod_{j=1}^{i} \prod_{k=1}^{t_j}  R_k(j,d_v,\rho(x)) \Big],\numberthis
\end{align*}
in which $R_k(j,d_v,\rho(x))$ denotes the number of possible basic trees of v--size $j$ in the ensemble that are rooted at the $k$th variable node out of $t_j$ variables nodes (of the cycle) whose basic trees all have v-size $j$. 
The terms $R_k(j,d_v,\rho(x))$ are random variables due to the randomness in the degree of check nodes participating in the basic trees. For different values of $k$, these random variables are independent and identically distributed, and have the following expected value:
\begin{equation} \label{eq019}
\mathbf{E}(R_k(j,d_v,\rho(x))) = (\overline{d_c}-1)^{j} \times B_{j+1}^{d_v},
\end{equation}
where $B_{j+1}^{d_v}$ was defined in the proof of Theorem~\ref{Th95}. To derive~(\ref{eq019}), we have used the fact that each check node that is a parent of one of the $j$ variable nodes in a basic tree of v-size $j$ has a degree selected randomly and independently from the distribution $\rho(x)$. Taking the expected value of (\ref{F3000}), replacing (\ref{eq019}) for $\mathbf{E}(R_k(j,d_v,\rho(x)))$ in the expected value, and considering the fact that random variables $R_k(j,d_v,\rho(x)))$ are independent for different values of $k$ and $j$, we obtain
\begin{align*}\label{eq018}
\mathbf{E}(N_{(a,a(d_v-2))}^{ETS} )&\sim\sum_{i=0}^{a-g/2} \sum_{{t_0, t_1,\ldots, t_{i} }\atop {(t_0, t_1,  \ldots, t_{i} \bigtriangledown a)}}\Big[ \mathbf{E}(N_{2(a-i)}) {{a-i}\choose{t_0, t_1,t_2, \ldots, t_{i} }} (\overline{d_c}-1)^{i} \times \prod_{j=1}^{i}  ( B_{j+1}^{d_v}) ^{t_j} \Big]\:.\numberthis
\end{align*} 
By replacing $\mathbf{E}(N_{2(a-i)})$ with (\ref{eq23}), $B_{j+1}^{d_v}$ with the upper bound of (\ref{eq01}), and using (\ref{F4}) for $C_{j+1}^{d_{v}}$ in (\ref{eq018}), we derive (\ref{eq200}). To prove the accuracy claimed in the theorem, we derive an asymptotic lower bound which is equal to the upper bound of (\ref{eq200}) multiplied by $[(a-g/2+1) \times S(h_w)]^{-a}$. The term $(a-g/2+1)^{-a}$ is due to lower bounding $B_{j+1}^{d_v}$ with the exact same steps as those taken in the proof of Theorem~\ref{Th95}. The term $S(h_w)^{-a}$ is due to lower bounding $\mathbf{E}(N_{2(a-i)})$ with $S(h_w)^{-(a-i)}$ times the asymptotic upper bound of (\ref{eq23}), as explained in Theorem~1 of~\cite{dehghan2016new}, and then taking the smallest value of $S(h_w)^{-(a-i)}$, by setting $i=0$.
\end{proof}

\subsection{Irregular bipartite graphs}
In irregular bipartite graphs, for a given value of $a$, the $(a,b)$ ETSs whose asymptotic multiplicity is a non-zero constant on average can have a variety of $b$ values. In Proposition~\ref{NT1}, we estimated the sum of all such ETSs over different values of $b$. In this subsection, we derive asymptotic estimates for the expected number of ETSs in a given $(a,b)$ class. For a given value of $a$, we start from the smallest $b$ value, i.e., $a(d_{v_{\min}} -2)$. This corresponds to ETSs that consist of a simple cycle and some trees connected to the variable nodes of that simple cycle, and all the variable nodes in the ETS have degree $d_{v_{\min}}$. This class is of particular interest since it is well-known that among $(a,b)$ ETSs with the same value of $a$, those with the smallest $b$ value are most harmful.

\begin{theo} \label{NT8}
Consider the random ensemble $\mathcal{G}$ of irregular bipartite graphs with variable and check node degree distributions $\lambda(x) =\sum_{i} \lambda_i x^{i-1}$ and $\rho (x) = \sum_{i} \rho_i x^{i-1}$, respectively, where
$\rho(x)$ corresponds to the set of check node degrees $\{d'_i\}_{i=1}^{n'}$. Let the number of nodes in the graphs of $\mathcal{G}$ tend to infinity.
In $\mathcal{G}$, for fixed values $a$, $d_{v_{\max}}$, $d_{c_{\max}}$ and $g$, we have
\begin{align*}\label{XF5}
\mathbf{E}( N_{(a,a(d_{v_{\min}}-2))}^{ETS}) &\approx \sum_{i=0}^{a-g/2} \sum_{{t_0, t_1,\ldots, t_{ i} }\atop {(t_0, t_1,  \ldots, t_{ i} \bigtriangledown a)}}\Big[ \displaystyle \Big(\frac{2\lambda_q(q-1)}{|E(G)|} \displaystyle\sum_{i=1}^{n'}{{d_i'}\choose {2}}\Big)^{(a-i)} \times \dfrac{1}{2(a-i)}\\
&\times{{a-i}\choose{t_0, t_1,t_2, \ldots, t_{ i} }} (\overline{d_c}-1)^{i} \times (\dfrac{\lambda_q/q}{\sum_{j}\lambda_j/j})^{i} \\
&\times \prod_{j=1}^{ i}  (\frac{ 1}{j+1}{{(q-1)(j+1) }\choose{j}})^{t_j} \Big]\:,\numberthis
\end{align*}
where $q = d_{v_{\min}}$, $|E(G)|$ is the number of edges in the graph $G \in \mathcal{G}$, and the approximation is an asymptotic upper bound within the fixed multiplicative factor of $[(a-g/2+1) \times S(h_w)]^{-a}$ from the exact value, with Specht's ratio $S(h)$ defined in Theorem~\ref{Th57},
%by $S(h)=\dfrac{(h-1)h^{\frac{1}{h-1}}}{e \log h}$ for $h \neq 1$, and $S(1)=1$ ($e$ is Euler's constant), 
and $h_w= \frac{d_{c_{\max}} (d_{c_{\max}}-1)}{d_{c_{\min}}(d_{c_{\min}}-1)}$.
\end{theo}

\begin{proof}
Using similar discussions as those in the proof of Theorem~\ref{Th95}, we obtain the following equation:
\begin{align*}\label{XF1}
N_{(a,a(d_{v_{\min}}-2))}^{ETS} &\sim \sum_{i=0}^{a-g/2} \sum_{{t_0, t_1,\ldots, t_{i} }\atop {(t_0, t_1,  \ldots, t_{i} \bigtriangledown a)}}\Big[ N_{2(a-i)}' {{a-i}\choose{t_0, t_1,t_2, \ldots, t_{i} }}  \times \prod_{j=1}^{i}  \prod_{k=1}^{t_j} R'_k(j,\lambda(x),\rho(x)) \Big],\numberthis
\end{align*}
where $ N_{2(a-i)}'$ is the number of cycles of length $2(a-i)$ in which the degree of each variable node is $q$, and $R'_k(j,\lambda(x),\rho(x))$ denotes the number of basic trees of v-size $j$, whose roots have degree $q-2$ and whose other variable nodes all have degree $q$, in the ensemble. The index $k$ is used to identify the root of the basic tree out of $t_j$ variables nodes (of the cycle) whose basic trees all have v-size $j$. 

Considering that each variable node in a basic tree has degree $q$ with probability $p_q$ given below:
$$
p_q = \dfrac{\lambda_q/q}{\sum_{j}\lambda_j/j}\:,
$$ 
one can see that 
$$
R'_k(j,\lambda(x),\rho(x)) = p_q^j R_k(j,q,\rho(x))\:,
$$
where $R_k(j,q,\rho(x))$ is defined in the proof of Theorem~\ref{NT3}, with its expected value calculated in (\ref{eq019}). We thus have
\begin{equation} \label{eq090}
\mathbf{E}(R'_k(j,\lambda(x),\rho(x))) = (\overline{d_c}-1)^{j} \times p_q^j \times B_{j+1}^{q}\:.
\end{equation}
Taking the expected value of (\ref{XF1}), and replacing $\mathbf{E}(R'_k(j,\lambda(x),\rho(x)))$ with (\ref{eq090}) in the expected value, we obtain
\begin{align*} \label{eq020}
\mathbf{E}( N_{(a,a(d_{v_{\min}}-2))}^{ETS}) &\sim \sum_{i=0}^{a-g/2} \sum_{{t_0, t_1,\ldots, t_{ i} }\atop {(t_0, t_1,  \ldots, t_{ i} \bigtriangledown a)}}\Big[ \mathbf{E}( N_{2(a-i)}'){{a-i}\choose{t_0, t_1,t_2, \ldots, t_{ i} }} (\overline{d_c}-1)^{i} \times p_q^i \times \prod_{j=1}^{ i}  (B_{j+1}^{q})^{t_j} \Big].\numberthis
\end{align*}
To calculate the expected value $\mathbf{E}( N_{2(a-i)}')$, we use the result of Theorem~\ref{Th57} and obtain
\begin{equation}\label{eq021}
\mathbf{E}( N_{2(a-i)}')\approx \displaystyle (\lambda_q(q-1))^{(a-i)} \displaystyle\Big( \frac{2}{|E(G)|} \sum_{i=1}^{n'}{{d_i'}\choose {2}}\Big)^{a-i}\dfrac{1}{2(a-i)}\:.
\end{equation} 
By replacing $\mathbf{E}(N_{2(a-i)}')$ with (\ref{eq021}), $B_{j+1}^{q}$ with the upper bound of (\ref{eq01}), and using (\ref{F4}) for $C_{j+1}^{q}$ in (\ref{eq020}), we derive (\ref{XF5}).

The proof for the accuracy of the estimate follows the same steps as those used in the proof of Theorem~\ref{NT3} for the accuracy of Estimate (\ref{eq200}).
\end{proof}

\begin{rem}
One can use (\ref{eq020}) along with (\ref{eq021}) and the exact values of $B_{j+1}^{q}$ (see Table~\ref{TAB1}), to obtain a more accurate estimate compared to (\ref{XF5}). 
\label{remsa}
\end{rem}
In an irregular ensemble, for a given value of $a$, ETSs in $(a,b)$ classes with $b$ in the range of  $(d_{v_{\min}} - 2)a \leq b  \leq (d_{v_{\max}}-2)a$ have finite non-zero asymptotic expected multiplicity. One can use an approach similar to that of Theorem~\ref{NT8} to estimate such expected values. The derivations, however, are in general tedious, as the calculations involve ETSs that are in the same class but consist of different combinations of variable nodes with different degrees.

\section{Numerical results}
\label{S5}

%\subsection{Variable-regular code ensembles}
In this section, we present some numerical results in relation to the theoretical results presented in previous sections. In particular, we provide the multiplicities of different trapping set structures within randomly constructed LDPC codes
and compare the results to the average values predicted by the asymptotic analysis. We first focus on LETS structures, and use the exhaustive search algorithms of~\cite{hashemi2015new} and~\cite{Y-Arxiv} to find the multiplicity of LETS structures within different classes.

\begin{table}[ht]
\caption{Biregular LDPC Codes from~\cite{Mackay} Used in the First Experiment}
\begin{center}
\scalebox{1}{
\begin{tabular}{ |c||c|c|c|c|c||c|c|c|c|  }
\hline
Code & $\mathcal{C}_1$ & $\mathcal{C}_2$ & $\mathcal{C}_3$ & $\mathcal{C}_4$ & $\mathcal{C}_5$ & $\mathcal{C}_6$ & $\mathcal{C}_7$ & $\mathcal{C}_8$& $\mathcal{C}_9$ \\
\hline
$n$     & 816 &1008  &4000  &20000 & 50000  &4000  & 8000 & 10000 & 20000\\
\hline
$R$     & 0.5 &0.5   & 0.5  & 0.5  &0.5     &0.5   & 0.5  &0.5   & 0.5 \\
\hline
$d_v$   & 3   &3     & 3    &3     &3       &4     & 4    &4     & 4\\
\hline
%girth   & 6   &6     & 6    &6     &6        &6    & 6    &6    & 6 \\
%\hline
\end{tabular}
}
\end{center}
\label{Tabel1}
\end{table}

For the first experiment, we consider random biregular $(3,6)$ and $(4,8)$ LDPC codes with different block lengths as listed in Table~\ref{Tabel1}.
The multiplicities of LETS structures in different $(a,b)$ classes with $a \leq 12, b \leq 5$, for the $(3,6)$ codes, and with $a \leq 10, \: b \leq 10$, for the $(4,8)$ codes, are listed in Tables \ref{T2} and~\ref{T3}, respectively.
In both tables, we have also listed the asymptotic average number of LETS structures within each class obtained based on Corollary~\ref{NEWL2}. For $(3,6)$ and $(4,8)$ codes, the asymptotic average value is non-zero only for classes with $b/a=1$ and $b/a=2$, respectively. For these cases, the asymptotic average multiplicity of $(a,a)$ and $(a,2a)$ classes is approximated by (\ref{eqsd}) as $10^a/(2a)$ and $21^a/(2a)$, respectively.

\begin{table}[ht]
\caption{Multiplicities of $(a,b)$ LETSs within the range $a\leq 12$, $b \leq 5$, for $(3,6)$ regular LDPC codes of Table~\ref{Tabel1}, in comparison with the asymptotic expected values of Corollary~\ref{NEWL2}}
\begin{center}
\scalebox{1}{
\begin{tabular}{ |c||c|c|c|c|c||c|  }
\hline
Code & $\mathcal{C}_1$ & $\mathcal{C}_2$ & $\mathcal{C}_3$ & $\mathcal{C}_4$ & $\mathcal{C}_5 $ & Expected value \\
\hline
\multicolumn{6}{|l|}{$(a,a)$} \\
\hline
(3,3) &    132    &   165     &  171    & 161    & 178   & 166.6 \\
\hline
(4,4) &    1491   &  1252     & 1219    & 1260   & 1268  & 1250 \\
\hline
(5,5) &    9169   &  10019    & 9935    & 10046  & 10231 & 10000\\
\hline
\multicolumn{6}{|l|}{$(a,b), b \neq a$} \\
\hline
(4,2) &    3    &   6     &  1    & 0     & 0      &0    \\
\hline
(5,3) &    90   &   100   &  21   & 2     & 1      &0   \\
\hline
(6,2) &    2    &   5     &  0    & 0     & 0      &0  \\
\hline
(6,4) &  2463   &  1885   &  476  & 95    & 52     &0   \\
\hline
(7,3) &  110    &   116   &  10   & 0     & 1      &0   \\
\hline
(7,5) &  33406  &  29736  &  7661 & 1540  & 640    &0   \\
\hline
(8,2) &  1      &   3     &  0    & 0     & 0      &0   \\
\hline
(8,4) &  4199   &  2961   & 183   & 4     & 1      &0   \\
\hline
(9,3) &  195    &   169   &  4    & 0     & 0      &0   \\
\hline
(9,5) & 84378   &  63787  & 4001  & 167   & 21     &0  \\
\hline
(10,2)&  15     &   4     &  0    & 0     & 0      &0 \\
\hline
(10,4) &  7965  &  4869   &  74   & 1     & 0      &0 \\
\hline
(11,3) &  290   &  219   &  1     & 0     & 0      &0  \\
\hline
(11,5) & 211273  &  134236  &2278 & 8     & 2      &0   \\
\hline
(12,2) &  15     &  6       &  0     & 0  & 0      &0 \\
\hline
(12,4) &  17838  &  9041    &  36   & 0   & 0      &0  \\
\hline
\end{tabular}
}
\end{center}
\label{T2}
\end{table}

\begin{table}[ht]
\caption{Multiplicities of $(a,b)$ LETSs within the range $a\leq 10$, $b \leq 10$, for $(4,8)$ regular LDPC codes of Table~\ref{Tabel1}, in comparison with the asymptotic expected values of Corollary~\ref{NEWL2}}
\begin{center}
\scalebox{1}{
\begin{tabular}{ |c||c|c|c|c|| c|  }
\hline
Code & $\mathcal{C}_6$ & $\mathcal{C}_7$   & $\mathcal{C}_8$   & $\mathcal{C}_9$   & Expected value \\
\hline
\multicolumn{5}{|l|}{$(a,2a)$} \\
\hline
(3,6) &    1563     &  1620   &  1598 & 1531  & 1543 \\
\hline
(4,8) &    24269    & 24107   &  24241& 24368 &24310 \\
\hline
(5,10) &    402513   & 406289 & 406754& 407743&408410 \\
\hline
\multicolumn{5}{|l|}{$(a,b), b \neq 2a$} \\
\hline
(4,6) &    91   & 55          & 33   &    17   & 0\\
\hline
(5,6) &    2   & 2            & 1    &    0    & 0\\
\hline
(5,8) &    4640   & 2303      & 1910 &   934   & 0\\
\hline
(6,8) &    588   & 196        & 107  &    22   & 0\\
\hline
(6,10) &    185544   & 96525  & 76621&    37075& 0 \\
\hline
(7,8) &    53   & 10          & 7    &    0    & 0 \\
\hline
(7,10) &    39794   & 10698   & 6649 &   1682  & 0 \\
\hline
(8,8) &    5   & 0            & 0    &    0    & 0\\
\hline
(8,10) &    7006   & 1116     & 567  &    60   & 0 \\
\hline
(9,8) &    1   & 0            & 0    &    0    & 0\\
\hline
(9,10) &    1145   & 101      & 53   &    0    & 0 \\
\hline
(10,10) &    185   & 13       & 4    &    0    & 0 \\
\hline
\end{tabular}
}
\end{center}
\label{T3}
\end{table}

The results of Tables~\ref{T2} and \ref{T3} show that the non-zero expected values provide a good approximation for the multiplicity of LETSs in $(a,a)$ and $(a,2a)$ classes, for $(3,6)$ and $(4,8)$ codes, respectively,
even at relatively short block lengths. The results also show that the multiplicity of other classes decrease with increasing $n$ (and tend asymptotically to zero). The rate of decrease, however, depends
on the class and is faster for some classes than others.

In the second experiment, we consider irregular LDPC codes with degree distributions $\lambda(x) = 0.4286x^2+0.5714x^3$, and $\rho(x) = x^6$, and construct random codes of
block lengths $n=500, 1000, 4000, 10000,$ and $20000$.  All the codes have rate $0.5$ and girth $6$. To obtain the asymptotic expected number of LETS structures for this ensemble, 
we use the results presented in Section~\ref{LETS-IRREG}. For any given $a$, we note that the asymptotic expected multiplicity is only non-zero for $b$ values in the range of 
$a \leq b \leq 2a$. We also note that the degree distribution of this ensemble matches the one discussed in Example~\ref{Ex29}. Using the notations of Example~\ref{Ex29}, for the ensemble under consideration,
we have $i=3$, $j=4$, $k=7$, $\alpha_j = b-a$ and $\alpha_i = 2a-b$. Replacing these in (\ref{eqex}), we obtain
 \begin{equation}
\mathbf{E}(N_{(a,b)}^{LETS}) \approx \displaystyle {{a}\choose{2a-b,b-a}}(2 \times 0.4286)^{2a-b} (3 \times 0.5714)^{b-a} \times \frac{6^{a}}{2a}\:.
\label{eqexx}
\end{equation}
Using (\ref{eqexx}), for $a=3, 4, 5, 6$, we obtain the results listed in the last column of Tables~\ref{Tab1}--\ref{Tab4}, respectively. In each table, the results for different values of $b$ in the range of $3 \leq b \leq 2a$, 
are provided. As explained before, the asymptotic multiplicities are non-zero only for $b \geq a$. In each table, we have also given the actual multiplicities of different $(a,b)$ LETSs for the five random codes. 
The comparison shows a good match between the theoretical asymptotic value and the finite-length numerical results even at relatively short block lengths.

%Based on Proposition~\ref{NT1}, the sum of the asymptotic expected values of LETSs of a given size $a$ tends to the average number of simple cycles of length $2a$. We have used the approximation~(\ref{eq12}) for the latter (i.e., $15.42^a/(2a)$), and have provided the values for different values of $a=3, 4, 5,$ and $6$, in the last column of Table~\ref{TTT1}. The actual sum of the multiplicities of LETSs of a given size $a$ for the five random codes are also provided in the table. The comparison shows a good match between the theoretical asymptotic value and the finite-length numerical results even at relatively short block lengths.

%\begin{table}[ht]
%\caption{Total multiplicities of LETSs of size $a$ for random irregular LDPC codes with degree distributions $\lambda(x) = 0.4286x^2+0.5714x^3$, $\rho(x) = x^6$, and different block lengths in comparison with the expected values from Proposition~\ref{NT1}}
%\begin{center}
%\scalebox{1}{
%\begin{tabular}{ |c||l|l|l|l|l||  c|  }
%\hline
%$a$ & \multicolumn{ 5}{  c| }{  Block Length}                     &   Expected values  \\
% \cline{2-6}
%        & 500  & 1000     &  4000   &10000   & 20000      & from Proposition \ref{NT1}\\
% \hline \hline
%3 & 651   & 598     & 617     & 604    &  600        & 611   \\
%4 & 7373  & 7173    & 7056    & 7055   &  7169       & 7067  \\
%5 &90134  & 89720   & 87564   & 87637  &  87630      & 87181 \\
%6 &1157736& 1150665 & 1130815 & 1127524&  1124863    & 1120275\\ \hline
%
%\end{tabular}
%}
%\end{center}
%\label{TTT1}
%\end{table}

\begin{table}[ht]
\caption{Multiplicities of LETSs of size three in the Tanner graphs of random irregular LDPC codes with degree distributions $\lambda(x) = 0.4286x^2+0.5714x^3$, $\rho(x) = x^6$, and with different block lengths in comparison with  the expected values from Equation (\ref{eqexx})}
\begin{center}
\scalebox{1}{
\begin{tabular}{ |c||l|l|l|l|l||  c|  }
\hline
Multiplicity  of LETSs  & \multicolumn{ 5}{  c| }{  Block Length}                     &   Expected values  \\
 \cline{2-6} for $3 \leq b \leq 6$ &        500  & 1000     &  4000   &10000   & 20000      & from Equation (\ref{eqexx})\\
 \hline \hline
$(3,3)$  & 26    & 22      & 18      & 22     &  31         & 22   \\
$(3,4)$  & 145   & 146     & 127     & 148    &  123        & 136  \\
$(3,5)$  &312    & 254     & 281     & 269    &  257        & 272 \\
$(3,6)$  &168    & 176     & 191     & 165    &  189        & 181\\ \hline

\end{tabular}
}
\end{center}
\label{Tab1}
\end{table}

\begin{table}[ht]
\caption{Multiplicities of LETSs of size four in the Tanner graphs of random irregular LDPC codes with degree distributions $\lambda(x) = 0.4286x^2+0.5714x^3$, $\rho(x) = x^6$, and with different block lengths in comparison with  the expected values from Equation (\ref{eqexx})}
\begin{center}
\scalebox{1}{
\begin{tabular}{ |c||l|l|l|l|l||  c|  }
\hline
$(4,b)$ LETS classes  & \multicolumn{ 5}{  c| }{  Block Length}                          &   Expected values  \\
 \cline{2-6} for $3 \leq b \leq 8$ &        500  & 1000     &  4000   &10000   & 20000        & from Equation (\ref{eqexx})\\
 \hline \hline
$(4,3)$  & 5     & 4       & 0       & 1      &  0          & 0   \\
$(4,4)$  & 130   & 102     & 72      & 95     &  93         & 87  \\
$(4,5)$  &841    & 726     & 735     & 731    &  713        & 699 \\
$(4,6)$  &2240   & 2122    & 2122    & 2075   &  2090       & 2098\\
$(4,7)$  &2884   & 2830    & 2743    & 2741   &  2871       & 2797 \\
$(4,8)$  &1273   & 1389    & 1384    & 1412   &  1402       & 1398\\ \hline

\end{tabular}
}
\end{center}
\label{Tab2}
\end{table}

\begin{table}[ht]
\caption{Multiplicities of LETSs of size five in the Tanner graphs of random irregular LDPC codes with degree distributions $\lambda(x) = 0.4286x^2+0.5714x^3$, $\rho(x) = x^6$, and with different block lengths in comparison with  the expected values from Equation (\ref{eqexx})}
\begin{center}
\scalebox{1}{
\begin{tabular}{ |c||l|l|l|l|l||  c|  }
\hline
$(5,b)$ LETS classes  & \multicolumn{ 5}{  c| }{  Block Length}                          &   Expected values  \\
 \cline{2-6} for $3 \leq b \leq 10$ &        500  & 1000     &  4000   &10000   & 20000        & from Equation (\ref{eqexx})\\
 \hline \hline
$(5,3)$   & 12    & 2       & 1       & 0      &  0          & 0   \\
$(5,4)$  & 143   & 58      & 11      & 6      &  1          & 0  \\
$(5,5)$   &1055   & 654     & 436     & 420    &  374        & 359 \\
$(5,6)$  &5549   & 4471    & 3649    & 3638   &  3657       & 3598  \\
$(5,7)$  &16731  & 15904   & 14657   & 14480  &  14412      & 14392\\
$(5,8)$  &29850  & 29089   & 29248   & 28661  &  28693      & 28780\\
$(5,9)$  &27041  & 28434   & 28390   & 28795  &  28909      & 28777\\
$(5,10)$ &9753   & 11108   & 11172   & 11637  &  11584      & 11509\\ \hline

\end{tabular}
}
\end{center}
\label{Tab3}
\end{table}

\begin{table}[ht]
\caption{Multiplicities of LETSs of size six in the Tanner graphs of random irregular LDPC codes with degree distributions $\lambda(x) = 0.4286x^2+0.5714x^3$, $\rho(x) = x^6$, and with different block lengths in comparison with  the expected values from Equation (\ref{eqexx})}
\begin{center}
\scalebox{1}{
\begin{tabular}{ |c||l|l|l|l|l||  c|  }
\hline
$(6,b)$ LETS classes  & \multicolumn{ 5}{  c| }{  Block Length}                          &   Expected values  \\
 \cline{2-6} for $3 \leq b \leq 12$ &        500  & 1000     &  4000   &10000   & 20000        & from Equation (\ref{eqexx})\\
 \hline \hline
$6,2)$  & 1     & 0       & 0       & 0      &  0          & 0   \\
$(6,3)$  & 10    & 0       & 0       & 0      &  0          & 0   \\
$(6,4)$  & 179   & 53      & 9       &  5     &  4          & 0  \\
$(6,5)$  &1771   & 717     & 132     & 75     &  18         & 0 \\
$(6,6)$   &11080  & 5840    & 2474    & 2015   &  1785       & 1542\\
$(6,7)$  &48495  & 32305   & 21896   & 19933  &  19267      & 18507\\
$(6,8)$  &142776 & 116436  & 99340   & 94744  &  93575      & 92527\\
$(6,9)$  &282304 & 265229  & 254057  & 247637 &  247219     & 246710\\
$(6,10)$ &352158 & 364261  & 371165  & 369248 &  368437     & 370022\\
$(6,11)$  &247006 & 276543  & 288721  & 295189 &  295669     & 295983\\
$(6,12)$  &71956  & 89281   & 93021   & 98685  &  98889      & 98649\\ \hline

\end{tabular}
}
\end{center}
\label{Tab4}
\end{table}

As the next experiment, we consider the ensemble of biregular $(3,6)$ LDPC codes and find the asymptotic average multiplicity of $(a,a)$ ETSs for different values of $3 \leq a \leq 5$. For a given $a$, the $(a,a)$ ETS class
is the only class with finite non-zero asymptotic multiplicity for this ensemble. We use (\ref{eq47}) of Theorem~\ref{Th95} for the derivations, and assume that the girth $g=6$.

For $a=3$, since $g=6$, in (\ref{eq47}), we have $a-g/2=0$. The outer summation therefore has only one term corresponding to $i=0$.
This implies that the inner summation has only one term that counts the expected number of cycles of length $6$, which by (\ref{eqsd}) is equal to $10^3/6=166.6$, asymptotically. We thus have
$$
\mathbf{E}(N_{(3,3)}^{ETS}) \sim 166.6\:.
$$

For $a=4$, the outer summation in (\ref{eq47}) has two terms corresponding to $i=0, 1$. For the case of $i=0$,  the inner summation is responsible for counting the expected number of chordless $8$-cycles,
which by (\ref{eqsd}) is equal to $10^4/8= 1250$, asymptotically. For the case of $i=1$, there is only one sequence of $t_0, t_1$ values that satisfies $(t_0, t_1 \bigtriangledown 4)$, and that is $t_0=2$ and $t_1=1$. 
For these values, the term inside the double-summation of (\ref{eq47}) is calculated as:
\begin{equation}\label{MR1}
\dfrac{ 2^{3}(6-1)^{4} }{2(3) }{{3}\choose{2,1}}   (\frac{1}{1}{{2 }\choose{0}})^{1}=2500\:.
\end{equation}
By adding the results for $i=0$ and $i=1$, we obtain
$$
\mathbf{E}(N_{(4,4)}^{ETS}) \sim 3750\:.
$$

For $a=5$, there are three terms in the outer summation involving the following calculations: (a) $i=0$; The inner summation counts the asymptotic expected number of chordless cycles of length $10$, which is $10^5/10=10000$.
(b) $i=1$; In this case, there is only one sequence $t_0, t_1$, that satisfies $(t_0, t_1 \bigtriangledown 5)$, and that is $t_0=3, t_1=1$. For this sequence, we have
\begin{equation}\label{MR2}
\dfrac{ 2^{4}(6-1)^{5} }{2(4) }{{4}\choose{3,1}}   (\frac{1}{1}{{2 }\choose{0}})^{1}=25000\:.
\end{equation}
(c) $i=2$; In this case, we have two possible sequences of $t_0, t_1, t_2$, satisfying $(t_0, t_1, t_2 \bigtriangledown 5)$, and they are $1, 2, 0$, and $2, 0, 1$.
These sequences have the following contributions, respectively:
\begin{equation}\label{MR3}
\dfrac{ 2^{3}(6-1)^{5} }{2(3) }{{3}\choose{1,2,0}}   (\frac{1}{1}{{2 }\choose{0}})^{2}=12500\:,
\end{equation}
\begin{equation}\label{MR4}
\dfrac{ 2^{3}(6-1)^{5} }{2(3) }{{3}\choose{2,0,1}}   (\frac{1}{2}{{4 }\choose{1}})^{1}=25000\:.
\end{equation}
By adding the results of Case (a),  (\ref{MR2}) for Case (b), and (\ref{MR3}) and (\ref{MR4}), for Case (c), we obtain 
$$
\mathbf{E}(N_{(5,5)}^{ETS}) \sim 72500\:.
$$

The asymptotic expected values of (\ref{eq47}), just derived, are listed along with the multiplicity of $(a,a)$ ETSs in Codes ${\cal C}_1-{\cal C}_5$ in Table~\ref{Ta2}. The comparison shows a good match between 
the two values across the range of block lengths. For the sake of completeness, for ${\cal C}_1-{\cal C}_5$, we present the non-zero multiplicities of the other ETS classes within the range $a\leq 7$, $b \leq 5$, in Table \ref{Ta3}. As expected, by increasing the block length the multiplicities decrease (and asymptotically tend to zero).

\begin{table}[ht]
\caption{Multiplicity of $(a,a)$ ETSs in biregular $(3,6)$ LDPC codes in comparison with the asymptotic expected value of (\ref{eq47})}
\begin{center}
\scalebox{1}{
\begin{tabular}{ |c||c|c|c|c|c||c|  }
\hline
$(a,a)$ ETS class & $\mathcal{C}_1$ & $\mathcal{C}_2$ & $\mathcal{C}_3$ & $\mathcal{C}_4$ & $\mathcal{C}_5$ & Expected value of (\ref{eq47}) \\
\hline
$(3,3)$   &  132   &165    &171   &161   & 178   &166.6  \\
\hline
$(4,4)$   & 3459  &3700   &3777  &3675  &3938   &3750  \\
\hline
$(5,5)$     & 67559  &70783  &72468 &71446 &75626  &72500 \\
\hline
\end{tabular}
}
\end{center}
\label{Ta2}
\end{table}

\begin{table}[ht]
\caption{Multiplicity of $(a,b)$ ETSs in biregular $(3,6)$ LDPC codes within the range $a\leq 7$, $b \leq 5$, and the asymptotic expected value from Theorem \ref{C2}}
\begin{center}
\scalebox{1}{
\begin{tabular}{ |c||c|c|c|c|c|c|  }
\hline
$(a,b)$ ETS class       & $\mathcal{C}_1$ & $\mathcal{C}_2$ & $\mathcal{C}_3$ & $\mathcal{C}_4$ & $\mathcal{C}_5$ & Expected Value \\
\hline
$(4,2)$ & 3   &6              &1                &0                & 0                &0   \\
\hline
$(5,3)$ & 120 & 160              &31               &2                & 1                &0   \\
\hline
$(6,2)$ & 2 & 5                &0                &0                & 0                &0   \\
\hline
$(6,4)$ & 4152 & 4080             &916              &125              & 67               &0   \\
\hline
$(7,3)$ & 130 & 163              &10               &0                & 1                &0   \\
\hline
$(7,5)$ &98465 & 97852            &23550            &3886             & 1899            &0   \\
\hline
\end{tabular}
}
\end{center}
\label{Ta3}
\end{table}  

As another example to compare the asymptotic theoretical results of this paper for ETSs with those of finite-length random codes, we consider the $(4,8)$-regular ensemble from which the codes ${\cal C}_6-{\cal C}_9$ in Table~\ref{Tabel1} 
are selected. Based on the results of Theorem~\ref{Th95}, the asymptotic expected multiplicity of $(a,2a)$ ETS classes for this ensemble can be approximated by (\ref{eq48}), or more accurately by (\ref{F8}). (This is the only class of $(a,b)$ ETSs with non-zero asymptotic expected multiplicity for different values of $b$.) In the following, we use (\ref{F8}) to estimate the multiplicity of $(a,2a)$ ETSs for $a=3, 4, 5$.

For $a=3$, $a-g/2=0$, and there is only one term corresponding to $i=0$, and $t_0 = 3$ in the double-summation of (\ref{F8}). This term
counts the asymptotic average number of cycles of length $6$ through (\ref{eqsd}), and is equal to $21^3/6= 1543.5$. We thus have
$$
\mathbf{E}(N_{(3,6)}^{ETS}) \sim 1543.5\:.
$$

For $a=4$, there are two terms in the outer summation: (a) $i=0$; In this case, the only term in the inner summation corresponds to $t_0=4$, and is equal to the asymptotic average number of chordless cycles of length $8$, obtained from (\ref{eqsd}) to be $21^4/8= 24310$. (b) $i=1$; In this case there is only one sequence $t_0=2, t_1=1$ that satisfies $(t_0, t_1 \bigtriangledown 4)$. The contribution of this sequence in (\ref{F8}), considering that $B_2^4=2$ from Table~\ref{TAB1}, is given by 
$$
\dfrac{ (d_v-1)^{3}(d_c-1)^{4} }{2(3) }{{3}\choose{2,1}}   (B_2^4)^{1}= 64827\:.
$$
By adding the results for $i=0$ and $i=1$, we obtain 
$$
\mathbf{E}(N_{(4,8)}^{ETS}) \sim 89137\:.
$$

For $a=5$, there are three terms in the outer summation of (\ref{F8}): (a) $i=0$; For this case, we have the contribution of $21^5/10= 408410$, which counts the asymptotic expectted number of chordless $10$-cycles.
(b) $i=1$; In this case there is only one sequence satisfying $(t_0, t_1 \bigtriangledown 5)$, i.e., $t_0=3$, $t_1=1$, with the following contribution:
\begin{equation}\label{MR2x}
\dfrac{ (d_v-1)^{4}(d_c-1)^{5} }{2(4) }{{4}\choose{3,1}}   (B_2^4))^{1}= 1361367\:.
\end{equation}
(c) $i=2$; In this case, we have two possible sequences $t_0, t_1, t_2$, satisfying $(t_0, t_1, t_2 \bigtriangledown 5)$. They are $1, 2, 0$, and $2, 0, 1$, with contributions:
%For $(1, 2, 0 \bigtriangledown 5)$ we should calculate the following:
\begin{equation}\label{MR3x}
\dfrac{ (d_v-1)^{3}(d_c-1)^{5} }{2(3) }{{3}\choose{1,2,0}}  (B_2^4)^{2}= 907578\:,
\end{equation}
%also, for $(2, 0, 1 \bigtriangledown 5)$  we should calculate the following (note that by (\ref{F02}), $B_3^4=7$).
and
\begin{equation}\label{MR4x}
\dfrac{ (d_v-1)^{3}(d_c-1)^{5} }{2(3) }{{3}\choose{2,0,1}}   (B_3^4 )^{1}= 1588261.5\:,
\end{equation}
respectively. (Note that from Table~\ref{Tabel1}, we have $B_3^4=7$.) By adding the results of Case (a), (\ref{MR2x}), (\ref{MR3x}) and (\ref{MR4x}), we obtain
$$
\mathbf{E}(N_{(5,10)}^{ETS}) \sim  4265616.5\:.
$$

The theoretical results, just derived, are listed in Table \ref{Ta4} along with multiplicities of ETSs in $(a,2a)$ classes of random codes. One can see a good match for different block lengths.

\begin{table}
\caption{Multiplicity of $(a,2a)$ ETSs in biregular $(4,8)$ LDPC codes in comparison with the asymptotic expected value of (\ref{F8})}
\begin{center}
\scalebox{1}{
\begin{tabular}{ |c||c|c|c|c|c|  }
\hline
$(a,2a)$ ETS class & $\mathcal{C}_5$  & $\mathcal{C}_6$ & $\mathcal{C}_7$ & $\mathcal{C}_8$ & Expected value of (\ref{F8})\\
\hline
$(3,6)$     &1563     &1620             &1598             & 1531            & 1543   \\
\hline
$(4,8)$     &89467    & 91891           & 91186           &88578            & 89137     \\
\hline
$(5,10)$    &4232030  & 4345623         &4325691          &4236433        & 4265616      \\
\hline
\end{tabular}
}
\end{center}
\label{Ta4}
\end{table}

We end this section with deriving approximations for the asymptotic expected multiplicity of $(a,a)$ ETSs for $a=3, 4, 5$, within the irregular ensemble with the following degree distributions: 
 $\lambda(x) = 0.4286x^2+0.5714x^3$, and $\rho(x) = x^6$. This is the same ensemble we discussed earlier in this section for LETS distributions. We use the results of Theorem~\ref{NT8}, and in particular Equations (\ref{eq020}) and (\ref{eq021}) to derive the expected values. (See Remark~\ref{remsa}.) We assume $g=6$, and note that for this ensemble $d_{v_{\min}} = 3$, and thus Equation (\ref{eq020}) can be used to find $\mathbf{E}( N_{(a,a)}^{ETS})$ for different values of $a$. For a given value of $a$, the $(a,a)$ ETS class has the smallest $b$ value among all the classes with non-zero asymptotic expected multiplicity for this ensemble. 
 
For $(3,3)$ ETSs, we have $a-g/2=0$, and the only term in the double-summation of (\ref{eq020}) is $\mathbf{E}(N_{6}')$, which is calculated by (\ref{eq021}) as
$(2 \times 0.4286)^3 \times 6^3 / 6 = 22.7$. We thus have
$$
\mathbf{E}( N_{(3,3)}^{ETS}) \approx 22.7\:.
$$
For $(4,4)$ ETSs, there are two term in the outer summation of (\ref{eq020}): (a) $i=0$; For this case, the inner summation is $\mathbf{E}(N_{8}')$, which can be calculated by (\ref{eq021})
as $(2 \times 0.4286)^4 \times 6^4 / 8 = 87.5$. (b) $i=1$; In this case there is only one sequence $t_0 =2, t_1=1$, that satisfies $(t_0, t_1 \bigtriangledown 4)$. The contribution of this sequence in the inner summation of (\ref{eq020}) is 
$$
\mathbf{E}(N_{6}') {{3}\choose{2,1}} (\overline{d_c}-1)^{1}   \times p_q^1 \times (B_{2}^3)^1\:.
$$
By substituting $\mathbf{E}(N_{6}') = 22.7$, $\overline{d_c}=d_c=7$, $p_q = 0.5$, and $B_{2}^3 = 1$, in the above equation,  we obtain the value $204$. 
By adding the results of Cases (a) and (b), we have
$$
\mathbf{E}( N_{(4,4)}^{ETS}) \approx 291.5\:.
$$
For $(5,5)$ ETSs, the outer summation of (\ref{eq020}) has three terms: (a) $i=0$; For this case, the inner summation is $\mathbf{E}(N_{10}')$, which by (\ref{eq021}) 
has a value $359$. (b) $i=1$; In this case, there is only one sequence $t_0=3, t_1=1$, that satisfies $(t_0, t_1 \bigtriangledown 5)$. The contribution of this sequence is
\begin{equation}
\label{eq2344}
\mathbf{E}(N_{8}' ){{4}\choose{3,1}} (\overline{d_c}-1)^{1} \times p_q^1 \times (B_{2}^3)^1 = 1049\:.
\end{equation}
(c) $i=2$; In this case we have two sequences $t_0, t_1, t_2$, that satisfy $(t_0, t_1, t_2 \bigtriangledown 5)$, and they are $1, 2, 0$, and $2, 0, 1$.
For the first sequence the contribution is
\begin{equation}\label{NR4}
\mathbf{E}(N_{6}') {{3}\choose{1,2,0}} (\overline{d_c}-1)^{2} \times p_q^2 \times (B_{2}^3)^2 \times (B_{3}^3)^0 = 612\:.
\end{equation}
For the second sequence, the contribution is
\begin{equation}\label{NR5}
\mathbf{E}(N_{6}'){{3}\choose{2,0,1}} (\overline{d_c}-1)^{2} \times p_q^2 \times  (B_{2}^3)^0 \times (B_{3}^3)^1 = 1224\:. 
\end{equation}
By adding the result of Case (a) with that of Case (b), i.e., (\ref{eq2344}), and those of Case (c), (\ref{NR4}) and (\ref{NR5}), we obtain
$$
\mathbf{E}( N_{(5,5)}^{ETS}) \approx 3244\:.
$$
The above results along with the multiplicity of corresponding ETS classes in randomly selected LDPC codes of different block length from the ensemble are listed in Table~\ref{Ta5}.

\begin{table}[ht]
\caption{Multiplicities of  $(a,a)$ ETSs in the Tanner graphs of  random irregular LDPC codes with degree distributions $\lambda(x) = 0.4286x^2+0.5714x^3$, and $\rho(x) = x^6$, and with different block lengths in comparison with the expected values from Theorem \ref{NT8}}
\begin{center}
\scalebox{1}{
\begin{tabular}{ |c||l|l|l|l|l||  c|  }
\hline
$(a,a)$ ETS   & \multicolumn{ 5}{  c| }{  Block Length}                      &   Expected values  \\
 \cline{2-6} classes         & 500  & 1000     &  4000   &10000   & 20000      & from Theorem \ref{NT8}\\
 \hline \hline
$(3,3)$                     & 26    & 22      & 18      & 22     &  31         & 22.7   \\
$(4,4)$                     & 341   & 287     & 216     & 272    &  304        & 291.5 \\
$(5,5)$                     &3979   & 3158    & 2340    & 2759   &  2883       & 3244 \\ \hline

\end{tabular}
}
\end{center}
\label{Ta5}
\end{table}
\section{Conclusion}\label{S7}

In this paper, we studied the asymptotic behavior of local structures within the Tanner graph of randomly constructed regular and irregular LDPC codes, as the code's block length tends to infinity.
Examples of such structures are different categories of trapping sets, such as stopping sets, elementary trapping sets, or absorbing sets, that are harmful in the error floor region of LDPC codes.
We derived a simple asymptotic relationship for the expected number of
such structures based on the difference between the number of nodes and the number of edges of the structure.
This was then related to the number of cycles in the structure, where we demonstrated that depending on the structure having, zero, one, or more than one cycle, the asymptotic expected value is infinity, a constant non-zero value, or zero, respectively. This general result was then applied to different categories of trapping sets to derive more specific results on the asymptotic expected values of different structures and classes of structures.
In particular, for the case where the asymptotic expected value of a structure is a non-zero constant, we derived the exact value or approximations for the constant value in terms of the girth and the degree distributions of the ensemble. For these derivations, we used the existing asymptotic results of \cite{dehghan2016new} on the average number of cycles. We also extended the results of \cite{dehghan2016new} to count the number of cycles with specific combination of variable and check node degrees. Moreover, to count the number of elementary trapping sets, we formulated the counting problem recursively, and solved it using a generalization of Catalan numbers. We also demonstrated through numerical results that the asymptotic results derived in this paper are rather accurate even at finite block lengths.

An important aspect of this work was to demonstrate that different (non-isomorphic) structures within the same class of trapping sets can, in general, behave differently, in the asymptotic regime of infinite block length. This was not investigated in previous studies where the focus had been on the asymptotic behavior of the whole class rather than its individual members. In addition, to the best of our knowledge, the results presented here, are the most general in the literature, in the sense that the framework developed in this work is applicable to any local substructure of a Tanner graph. Previous works were limited to specific categories of substructures such as stopping sets, elementary trapping sets or absorbing sets. In fact, the results presented here for leafless elementary trapping sets, which are the vast majority of dominant trapping sets over the AWGN channel and the BSC, are the first of their kind. We also believe that the results presented in this work on the cases where the asymptotic expected multiplicity of a structure is a non-zero constant are the first of their kind, in the sense that, while the fact that the asymptotic average multiplicity of a certain class of structures tends to a constant was reported in the existing literature, the exact value of that constant was not determined. In this work, the exact value or an estimate for such constants is provided as a function of girth and degree distributions of the ensemble.

\end{document}